\documentclass[journal,10pt]{IEEEtran}

\usepackage{cite}
\usepackage{color}
\usepackage{stfloats,tabularx}

\newcolumntype{L}[1]{>{\raggedright\arraybackslash}m{#1}}

\makeatletter
\renewcommand\normalsize{%
	\@setfontsize\normalsize\@xpt\@xiipt
	\abovedisplayskip 0.8\p@ \@plus1.5\p@ \@minus3.5\p@
	\abovedisplayshortskip 0.8\p@ \@plus1.5\p@
	\belowdisplayshortskip 0.8\p@ \@plus1.2\p@ \@minus2\p@
	\belowdisplayskip \abovedisplayskip
	\let\@listi\@listI}

\ifCLASSINFOpdf
   \usepackage[pdftex]{graphicx}

\else

\fi

\usepackage{cite}
\usepackage{amsmath,amsthm,amssymb,graphicx,algorithm,algorithmic}
\usepackage{makecell}
\usepackage{subfigure}
\usepackage{float}
\usepackage{hyperref}
\usepackage{booktabs}

\usepackage{enumerate}

\newtheorem{theorem}{\bf Theorem}

\newtheorem{lemma}{\bf Lemma}

\newtheorem{remark}{\bf Remark}

\usepackage{color}
\usepackage{array}
\usepackage{multicol}

\begin{document}

\title{Robust Transmission Design for RIS-Assisted High-Speed Train Communication Coverage Enhancement With Imperfect Cascaded Channels}

\author{

  Changzhu~Liu,~\IEEEmembership{Graduate~Student~Member,~IEEE},
  Ruisi~He,~\IEEEmembership{Senior~Member,~IEEE},
  Haoxiang~Zhang, \\
  Jiahui~Han,
  Ruifeng~Chen,
  Bo~Ai,~\IEEEmembership{Fellow,~IEEE},
  and Zhangdui~Zhong,~\IEEEmembership{Fellow,~IEEE}

\thanks{

Changzhu Liu, Ruisi He, Bo Ai, and Zhangdui Zhong are with the School of Electronics and Information Engineering, Beijing Jiaotong University, Beijing 100044, China (e-mails: changzhu{\_}liu@bjtu.edu.cn; ruisi.he@bjtu.edu.cn; boai@bjtu.edu.cn; zhdzhong@bjtu.edu.cn).

Haoxiang Zhang and Jiahui Han are with the China Academy of Industrial Internet, Ministry of Industry and Information Technology, Beijing, China (zhx61778294@126.com; hjh1760708@126.com).

Ruifeng Chen is with the Institute of Computing Technology, China Academy of Railway Sciences Corporation Limited, Beijing 100081, China (e-mail: ruifeng{\_}chen@126.com).

              }
              }

{}

\maketitle

\begin{abstract} %
  Reconfigurable intelligent surface (RIS) has recently been gained attention as an effective technique improving the coverage and performance of communication systems by creating additional communication links. Deployment of RIS is crucial for overcoming signal coverage limitations, especially in high-speed train (HST) scenarios. Considerable research has been performed assuming perfect channel state information (CSI). However, due to the rapidly time-varying fading channels and feedback delays, achieving perfect CSI at the base station (BS) is not feasible in the HST scenarios. To tackle this problem, this paper investigates a robust design strategy for RIS-aided HST communication coverage enhancement, particularly focusing on cascaded BS-RIS-user channels at BS (CBRUB). The study explores the optimization problem under two types distinct of models: centered on minimizing transmit power subject to worst-case rate constraints within the bounded CSI error (BCSIE) model, and the other focusing on outage probability (OP) constraints under the statistical CSI error (SCSIE)  model. We use the S-procedure to  approximate  the non-convex (NC) constraints, converting the worst-case rate constraints into linear matrix inequalities. Additionally, the Bernstein-type inequality is applied to transform the OP constraints into second-order cone constraints and linear inequalities. The simulation analysis results show that CBRUB errors have a significant effect on system performance compared to direct CSI errors.
\end{abstract}

%
\IEEEpeerreviewmaketitle

\section{Introduction}

\IEEEPARstart{O}{ver} the past decade, high-speed train (HST) has become a dominant mode of transportation, benefiting from superior safety, high speeds, and large passenger capacities, enabling the simultaneous transport of thousands of passengers \cite{r1new1,r1new2,r1new3,r1new4}. With the increasing demand for reliable connectivity, optimizing wireless network performance inside train carriages has become crucial to ensuring consistent communication services for all passengers, regardless of their location within the train \cite{r2}. The high mobility of HST presents unique challenges to wireless communication systems, resulting in a complex environment characterized by time-varying channel estimation, coverage expansion, channel modeling, and beamforming design \cite{r3, r3new1, r3new2, r3new3, r3new4,r3new5,r3new6}. Consequently, signal prediction and wireless network planning for HST communications face substantial difficulties. In this context, reconfigurable intelligent surface (RIS) has attracted significant attention as a highly promising and cost-effective solution for enhancing communication performance across various applications \cite{r4,r4new1,r4new2,r4new3,r4new4}. RIS technology enables intelligent optimization the radio environment and is distinguished by its low cost, minimal complexity, and ease of deployment, making it particularly suitable for high-mobility communication systems \cite{r5}.

Recent studies on RIS-aided HST communications have highlighted the potential of RIS to significantly improve wireless coverage and system performance through the joint optimization of base station (BS) beamforming and RIS reflection beamforming \cite{r6,r7,r8,r9}. However, the algorithms proposed in these studies typically assume that perfect channel state information (CSI) at the BS (CSIB). Such an assumption is highly idealistic in practical HST scenarios, where the CSI is inevitably imperfect and outdated, leading to non-negligible performance degradation, particularly for RIS-assisted links that rely on coherent phase alignment.

A key reason why perfect CSI fails in RIS-assisted HST systems lies in the combined effects of Doppler-induced channel variation, phase noise, and feedback delays \cite{r3,r10,r11}. First, due to the high train speed, the channel experiences pronounced Doppler shifts, which substantially shorten the channel coherence time. As a consequence, even if CSI is accurately estimated at time instant $t$, the effective channel during data transmission at $t+\tau$ can differ significantly, where $\tau$ includes pilot transmission, channel estimation, baseband processing, scheduling, and feedback latency. This mismatch, known as channel aging, fundamentally invalidates the assumption of instantaneous CSI and results in beamformers being optimized for an outdated channel realization \cite{r12}. Second, practical transceivers suffer from phase noise due to oscillator imperfections and hardware impairments \cite{r11}. In high-mobility and high-frequency operations, phase noise introduces random time-varying phase rotations and additional distortions, which not only deteriorate channel estimation quality but also disrupt coherent combining. This effect is particularly critical for RIS-assisted transmission because RIS beamforming hinges on accurate phase alignment across numerous reflecting elements; even small phase perturbations can accumulate across elements, collapse the intended passive beamforming gain, and cause unintended interference leakage \cite{r13}. Third, CSI acquisition requires a closed-loop procedure (estimation, quantization, and feedback), resulting in unavoidable feedback delays and signaling overhead \cite{r10,r20,r21}. In high-mobility scenarios, these delays further exacerbate channel aging, while the overhead limits how frequently CSI and RIS configurations can be updated \cite{r14,r15new}. Moreover, the nearly passive nature of RIS means it cannot actively transmit/receive pilots, which makes RIS-related CSI acquisition intrinsically difficult and often training-intensive, thereby amplifying the effective delay and mismatch. Therefore, assuming perfect CSI (implicitly assuming negligible Doppler variation, negligible phase noise, and near-zero feedback delays/overhead) is unrealistic for RIS-assisted HST communications, and it is essential to develop transmission designs that explicitly account for imperfect CSI  in HST communication systems.

RIS-assisted HST systems involve two types of channels: one directly connecting the BS to the user, and another related to RIS. The direct channel is relatively straightforward to estimate using conventional techniques such as the least squares algorithm. Consequently, the majority of prior research has concentrated on estimating the BS-RIS and the RIS-user channels.

Typically, there are two primary strategies for estimating RIS-associated channels. The first one involves directly estimation of the RIS-related channels, namely the BS-RIS and RIS-user channels \cite{r15}. As detailed in \cite{r15}, this can be achieved by integrating active channel elements into the RIS to facilitate channel estimation. However, this technique comes with several drawbacks. The integration of active elements leads to higher hardware costs and extra power consumption, causing an unsustainable burden on the RIS. Moreover, the estimated channel information must be sent back to the BS, resulting in increased signaling overhead and feedback latency, which is particularly undesirable in fast time-varying HST scenarios.

It has been found that the cascaded channels involving the BS, RIS, and user, formed by the product between the BS-RIS and RIS-user channels, are adequate for the development of joint active and reflection beamforming \cite{r16,r17,r18,r19}. This observation has led most of the current studies to focus on the second method, which involves estimating the cascaded channels \cite{r20,r21,r22}. In particular, channel estimation for cascaded channels has been investigated in both single-user multiple-input multiple-output (SU-MIMO) systems \cite{r20} and multi-user multiple-input single-output (MU-MISO) systems \cite{r21}. However, the estimation approaches in \cite{r20,r21} suffer from excessively high pilot overhead, which grows with the number of RIS reflection elements. To address this issue, the authors in \cite{r22} proposed a compressed sensing-based method that exploits the sparsity and correlation inherent in multiuser cascaded channels. Furthermore, the authors in \cite{r23} introduced an alternative method to reduce pilot overhead by employing auxiliary variables to decompose the received signal model into a subcarrier-wise bilinear sub-model and two linear sub-models with respect to the cascaded BS-RIS-user channels.

The aforementioned studies \cite{r6,r7,r8,r9} did not account for the impact of CSI imperfection on the system performance. In practical HST scenarios, channel estimation errors are inevitable, and the CSI mismatch is further aggravated by Doppler-induced channel aging, feedback delays, and phase noise. Consequently, treating the estimated CSI as perfect can lead to severe performance degradation, particularly for RIS-assisted links where coherent reflection gains are highly sensitive to phase inaccuracies. Therefore, developing robust transmission strategies for RIS-assisted HST communication systems that explicitly account for imperfect CSI becomes essential to bridge the gap between idealized theoretical gains and reliable performance in realistic high-mobility deployments.

To the best of our knowledge, there are some studies have investigated robust transmission of HST \cite{r24,r25,r26,r27}. Specifically, \cite{r24} illustrated how directional antennas and unmanned aerial vehicles can enhance the robustness of HST communication systems. In \cite{r25}, the authors  proposed a method  to improve robustness of HST communication system by leveraging the existing equipment to assist relaying.  Meanwhile, in \cite{r26}, the authors employed a priority weighted approach and developed a robustness-enhancing method based on graph theory. RIS was considered in \cite{r27} to improve the robustness of HST communication systems. However, these works still assume perfect CSI and do not explicitly address the fundamental challenges brought by Doppler-induced channel aging, phase noise, and CSI feedback delays in practical RIS-assisted HST scenarios. This gap motivates robust RIS-aided HST transmission design under imperfect CSI.

Against the above challenge, this paper investigates the robust transmission strategies of HST communication coverage enhancement, taking into account imperfect cascaded BS-RIS-user channels at the BS (CBRUB). We aim  to develop a robust design for both active and reflection beamforming that minimizes the transmit power accounting for both bounded CSI error (BCSIE) and statistical CSI error (SCSIE) models. The contributions of this work are outlined as follows:
\begin{itemize}
\item As far as we know, this is the first study that examines the robust transmission design for HST communication coverage enhancement by considering imperfect cascaded BS-RIS-user channels, which is more realistic than the previous works that assumed perfect CSI. Moreover, we focus on the robust design involves two different channel error models: the BCSIE model and the SCSIE model.
\item Firstly, we study worst-case robust beamforming design for the BCSIE model to minimize the transmit power, subject to the unit-modulus constraint for reflection beamforming and  worst-case quality-of-service (QoS) requirements under imperfect CBRUB. This approach ensures that each user's achievable rate satisfies its minimum requirement, regardless of the possible channel error realizations. To solve the non-convex (NC) problem, the S-procedure is used to approximate the semi-infinite inequality constraints. Then, within the framework of alternate optimization (AO), the active beamformer is refined using second-order cone programming (SOCP), while the passive beamforming is optimized using the penalty convex-concave procedure (CCP).
\item For the SCSIE model, we aim to minimize the transmit power under to unit-modulus and outage probability (OP) constraints. Specifically, the OP constraints guarantee that the probability of a user's achievable rate falling below its minimum requirement remains less than a predefined threshold. To solve this, we utilize the Bernstein-Type Inequality (BTI) to safely approximation the OP. The optimization of both the active beamforming and the passive beamforming are then solved iteratively using semidefinite relaxation (SDR).
\item Simulation results demonstrate that robust beamforming based on SCSIE model achieves superior system performance compared to the BCSIE model, particularly in terms of complexity, convergence speed, and transmit power. Furthermore, we find that the CBRUB error level plays a critical role in RIS-aided HST communication coverage enhancement. Specifically, when the CBRUB error is small, increasing the number of RIS reflecting elements reduces transmit power due to improved beamforming gain. Whereas, when the CBRUB error is large, increasing the number of RIS elements raises transmit power because of increased channel estimation error. Therefore, the effectiveness of deploying RIS depends on the CBRUB error level.
\end{itemize}

The organization of this paper is as follows. The system model and channel uncertainty are discussed in Section \ref{sy}. In Section \ref{wcrbd}, the formulation and solution of problems with the worst-case robust design  are presented. Section \ref{ocrbd} addresses the formulation and resolution of problems with outage-constrained robust design. Section \ref{complexity} examines the computational complexity of the proposed algorithms. Lastly, Sections \ref{nr} and \ref{con} present the numerical simulation results and conclusions, respectively.

{\textbf{\textit{Notations}}}: In this paper, scalar quantities are indicated by italicized letters, whereas bold lowercase and uppercase characters denote vectors and matrices, respectively. $\mathbb{C}$ denotes complex field. The Frobenius norm, conjugate, transpose, and conjugate transpose of matrix $\mathbf{X}$ are denoted by $||\mathbf{X}||_{F}$, $\mathbf{X}^{*}$, $\mathbf{X}^\mathrm{T}$, and  $\mathbf{X}^{\mathrm{H}}$,  respectively. The notation $||\mathbf{x}||_{2}$ represents the 2-norm of vector $\mathbf{x}$. $\left[\mathbf{X}\right]_{m,n}$ refers to the element located at the $m$th row and $n$th column of matrix $\mathbf{X}$, and $\mathbf{X} \succeq \mathbf{0}$ represents that $\mathbf{X}$ denotes a semi-definite matrix. The symbol $\left\lvert \cdot \right\rvert$ is the absolute value of its argument. $\mathrm{Tr}\{\cdot\}$, $\mathfrak{R}  \left\{ \cdot \right\} $, and $\lambda(\cdot)$ denote the trace, real part, modulus, and eigenvalue, respectively. The operation $\mathrm{diag} \left\{ \cdot\right\}$  denotes the diagonalization of a matrix. The symbol $\otimes$ represents the Kronecker product.  $\mathbf{I}$ denotes the identity matrix. Finally, $\mathcal{C} \mathcal{N} (\mathbf{0},\mathbf{I})$  refers to a random vector that follows a circularly symmetric complex Gaussian (CSCG) distribution, with zero mean and covariance matrix $\mathbf{I}$.

\section{System Model And Channel Uncertainty} \label{sy}
\subsection{System Model}
As depicted in F{}ig.~\ref{fig:1}, we analyze a RIS-assisted MU-MISO HST communication system, where a BS equipped with $M$ antennas communications with $K$ single-antenna users assisted by a RIS equipped with $N$  elements. The BS is a rectangular planar array (URA) with $M = M_{\mathrm{h}} \times M_{\mathrm{v}}$ elements, where $M_{\mathrm{h}}$ and $ M_{\mathrm{v}}$ are the numbers of elements along the horizontal and vertical axes, respectively. Similarly, the RIS is a URA  with $N =  N_{{\mathrm{h}}} \times N_{{\mathrm{v}}}$  elements. Let $\mathbf{s}=[ s_1,\cdots ,s_K ] \in \mathbb{C} ^{K\times 1}$ denote transmits $K$ Gaussian data symbols to each user with $\mathbb{E} \left\{ \mathbf{ss}^{\mathrm{H}} \right\} =\mathbf{I}$. For convenience, we define the sets $\mathcal{K} = \{ 1,\cdots ,K\} $ and $\mathcal{N} = \{ 1,\cdots ,N \}$ to represent the indices of the users and RIS elements, respectively. Thus, the signal received by user $k$ from the BS can be expressed as
\begin{equation}
  y_k=( \mathbf{h}_{\mathrm{D},k}^{\mathrm{H}}+\mathbf{h}_{\mathrm{R},k}^{\mathrm{H}}\mathbf{\Theta G} ) \mathbf{Ws}+n_k,\forall k\in \mathcal{K}.
\end{equation}
Here, the beamforming matrix is denoted as $\mathbf{W}=\left[ \mathbf{w}_1,\dots ,\mathbf{w}_K \right] \in \mathbb{C} ^{M\times K}$, where $\mathbf{w}_k\in \mathbb{C} ^{M\times 1}$ represents the beamforming vector for user $k$. The BS transmit power is given as by $\mathbb{E} \{ \mathrm{Tr}[ \mathbf{Wss}^{\mathrm{H}}\mathbf{W}^{\mathrm{H}} ] \} =\left\| \mathbf{W} \right\| _{F}^{2}$. $n_k$ denotes the additive white Gaussian noise (AWGN) at the $k$th user, with zero mean and variance $\sigma_k^2$, i.e., $n_k\sim \mathcal{C} \mathcal{N} (0,\sigma _{k}^{2})$. The RIS reflection beamforming  is represented by a diagonal matrix $\mathbf{\Theta }_i=\mathrm{diag}\{ \theta _1,\cdots ,\theta _N \} \in \mathbb{C} ^{N\times N}$, of which has unit modulus phase shifts, i.e., $\left\lvert \theta_n\right\rvert ^2=1$. The channel matrix between the BS and the RIS is denoted as $\mathbf{G}$, the channel vectors between the BS and user $k$, and between the RIS and user $k$, are denoted by $\mathbf{h}_{\mathrm{D},k}$ and $\mathbf{h}_{\mathrm{R},k}$, respectively. The main symbols of the paper are listed in Table \ref{table-notations}.
\begin{figure}[!t]
  \centering
  {\includegraphics[scale=0.25]{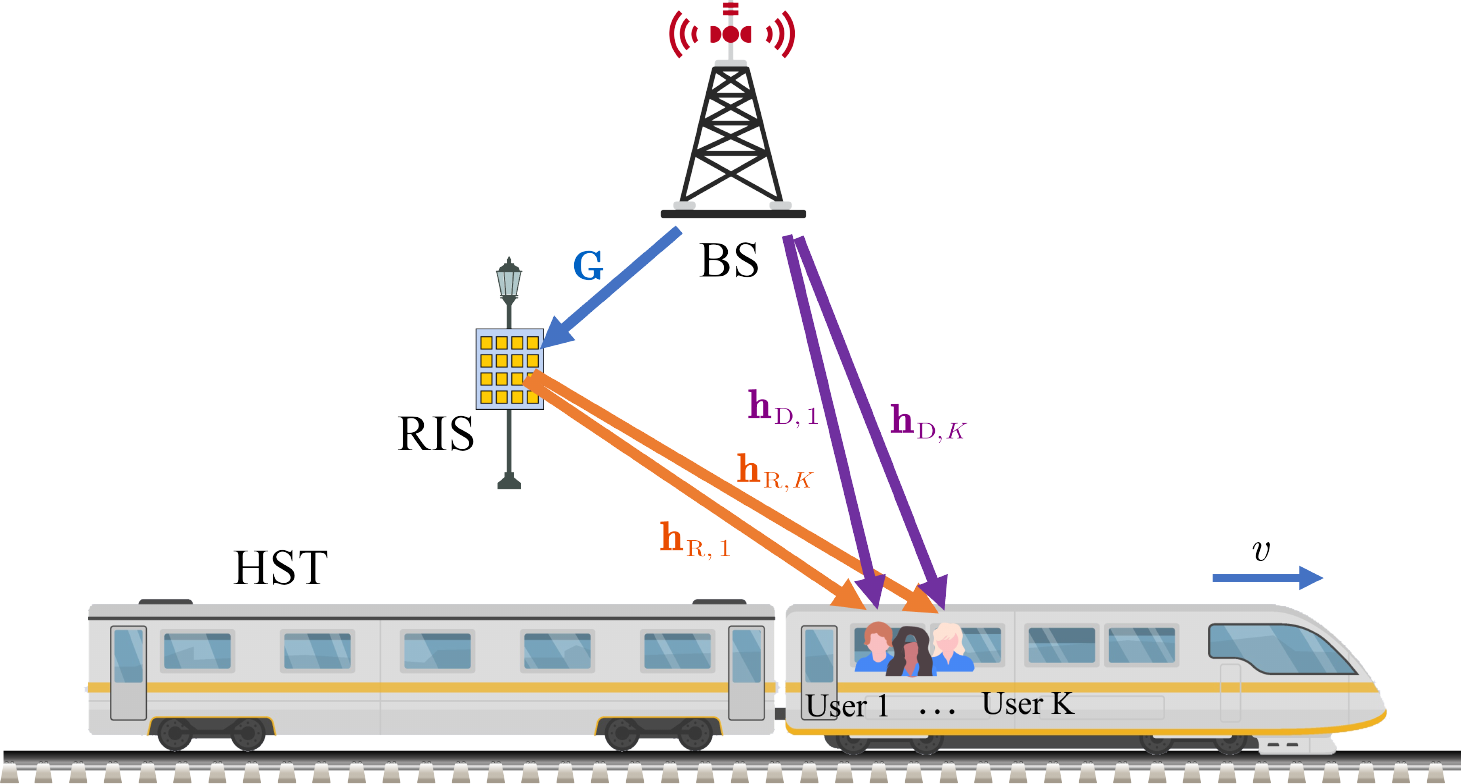}}
  \caption{ \label{fig:1} RIS-assisted HST communication model.}
\end{figure}

\begin{table}[!t]
  \centering
  \caption{NOTATIONS}
  \label{table-notations}
  \begin{tabularx}{\linewidth}{|L{3.2cm}|X|}
  \hline
  \textbf{Notation} & \textbf{Description} \\
  \hline
  $K$ & The number of  users \\
  \hline
  $M/N$ & The number of antennas each BS and RIS \\
  \hline
  $s_k$ & Transmission symbol of the $k$-th user \\
  \hline
  $\mathbf{h}_{\mathrm{D},k}$ & Direct channel from the BS to the $k$-th user \\
  \hline
  $\mathbf{h}_{\mathrm{R},k}$ & The channel from the RIS to the $k$-th user \\
  \hline
  $\mathbf{G}$ & The channel from the BS to the RIS \\
  \hline
  $\mathbf{W}/\mathbf{w}_k $ &  Beamforming matrix, beamforming vector of the $k$-th user \\
  \hline
  $n_k$ & Complex AWGN of the $k$-th user \\
  \hline
  $v$ & Speed of HST \\
  \hline
  $f_c$ & Carrier frequency \\
  \hline
  $\lambda_c$ & Carrier wavelength \\
  \hline
  $\sigma_{k}^2$ & Variances of $n_k$ \\
  \hline
  $\mathbf{\Theta}/\boldsymbol{\theta }$ & RIS phase shift matrix, vector \\ 
  \hline
  $ \theta_n$ & The phase shift of the $n$th RIS reflecting element \\ 
  \hline
  $\mathbf{H}_k$ & The cascaded channel from the BS to user $k$ via the RIS \\ 
  \hline
  $R_k$ & Achievable data rate of the $k$-th user \\ 
  \hline
  $\kappa$ & Rician K-factor \\ 
  \hline
  $PL_{\mathrm{D},k}/PL_{\mathbf{G}}/PL_{\mathrm{R},k}$ & Path losses \\ 
  \hline
  $\overline{\mathbf{h}}_{D,k}/\overline{\mathbf{G}}/\overline{\mathbf{h}}_{\mathrm{R},k}$ & LoS components \\ 
  \hline
  $\widetilde{\mathbf{h}}_{D,k}/\widetilde{\mathbf{G}}/\widetilde{\mathbf{h}}_{\mathrm{R},k}$ & NLoS components \\ 
  \hline
  $\beta _{\mathrm{D},k}/\beta _{\mathbf{G}}/\beta _{\mathrm{R},k}$ & Path loss exponents \\ 
  \hline
  $\phi ^{\{ {\mathrm{D},k};\mathrm{R};{\mathrm{R},k}  \}  } /\delta ^{\{ {\mathrm{D},k};\mathrm{R};{\mathrm{R},k}  \} }$ & Azimuth and elevation AoD \\ 
  \hline
  $\phi ^{\mathrm{B}}/\delta ^{\mathrm{B}}$ & Azimuth and elevation AoA \\ 
  \hline
  $\mathbf{a}_{\mathrm{B}}/\mathbf{a}_{\mathrm{R}}$ & Steering vector at the BS and the RIS\\ 
  \hline
  $f_{d_1}/f_{d_2}$ & Doppler frequency shift\\ 
  \hline
  $\widehat{\mathbf{H}}_{\mathrm{D},k}/\bigtriangleup \mathbf{H}_{\mathrm{D},k}$ & The estimated cascaded CSI and error vectors\\ 
  \hline
  $\widehat{\mathbf{h}}_{\mathrm{D},k}/\bigtriangleup \mathbf{h}_{\mathrm{D},k}$ & The estimated cascaded DCSIB error and error vectors\\ 
  \hline
  $\xi _{\mathrm{H},k}/\xi _{\mathrm{D},k}$ & The radii of the uncertainty areas\\ 
  \hline
  $\mathbf{\Sigma }_{\mathrm{H},k} /\mathbf{\Sigma }_{\mathrm{D},k}$  & Covariance matrices of error \\ 
  \hline
  $R_{\mathrm{th}}$  & Target rate \\ 
  \hline
  $\epsilon _{\mathrm{stop}}$  & Maximum tolerance value \\ 
  \hline
  $ \varepsilon _k$  & Maximum outage probability of $k$th user \\ 
  \hline
  $\omega_{\mathrm{H}}/\omega_{\mathrm{D}}$  & CSI uncertainty level\\ 
  \hline
  \end{tabularx}
\end{table}

Define the cascaded channel from the BS to the $k$th via the RIS as $\mathbf{H}_k=\mathrm{diag}( \mathbf{h}_{\mathrm{R},k}^{\mathrm{H}} ) \mathbf{G}$. Let  $\boldsymbol{\theta }=[ \theta _1,\dots ,\theta _N ] ^{\mathrm{T}}\in \mathbb{C} ^{N\times 1}$ be the vector containing diagonal elements of the matrix $\mathbf{\Theta}$. Moreover, $\iota _k=\| ( \mathbf{h}_{\mathrm{D},k}^{\mathrm{H}}+\boldsymbol{\theta }^{\mathrm{H}}\mathbf{H}_k ) \mathbf{W}_{-k} \| _{2}^{2}+\sigma _{k}^{2}$ is the interference-plus-noises (INs) power for user $k$, where $\mathbf{W}_{-k}=\left[ \mathbf{w}_1,\dots ,\mathbf{w}_{k-1},\mathbf{w}_{k+1},\dots ,\mathbf{w}_K \right] $. Therefore, the achievable data rate of the $k$th user is expressed by
\begin{equation} \label{dr}
  R_k( \mathbf{W},\boldsymbol{\theta } ) =\log _2( 1+\frac{1}{\iota _k}| ( \mathbf{h}_{\mathrm{D},k}^{\mathrm{H}}+\boldsymbol{\theta }^{\mathrm{H}}\mathbf{H}_k ) \mathbf{w}_k |^2 ).
\end{equation}

\subsection{Channel Model}
We utilize the Rician fading model to represent all the relevant channels.

\subsubsection{BS-user Channel}
$\mathbf{h}_{\mathbf{D},k}\in \mathbb{C} ^{M\times 1}$  is given by
\begin{equation}
  \mathbf{h}_{\mathrm{D},k}=\sqrt{PL_{\mathrm{D},k}}\left( \sqrt{\frac{\kappa}{\kappa+1}}\overline{\mathbf{h}}_{D,k}+\sqrt{\frac{1}{\kappa+1}}\widetilde{\mathbf{h}}_{\mathrm{D},k} \right),
\end{equation}
where $\kappa \geq 0 $ is the Rician K-factor, $PL_{\mathrm{D},k}$ denotes the path loss, which can be given in logarithmic form by 
\begin{equation}
  PL_{\mathrm{D},k}=PL_0-10\beta _{\mathrm{D},k}\log _{10}( \frac{d_0 }{d_{\mathrm{D},k} }   ),
\end{equation}
where $PL_0$ refers the path loss based on the reference distance of $d_0$, $\beta _{\mathrm{D},k}$ denotes the path loss exponent, and $d_{\mathrm{D},k}$ denotes the distance between the BS and the $k$th user. $\overline{\mathbf{h}}_{\mathrm{D},k}\in \mathbb{C} ^{M\times 1}$ is the line-of-sight (LoS) component, which is given by 
\begin{equation}
  \overline{\mathbf{h}}_{\mathrm{D},k}=e^{j2\pi f_{d_1}\tau}\mathbf{a}_{\mathrm{B}}( \phi ^{\mathrm{D},k},\delta ^{\mathrm{D},k} ),
\end{equation}
where $f_{d_1}=\upsilon \cos \phi ^{\mathrm{D},k}\cos \delta ^{\mathrm{D},k}/\lambda _c$ denotes the Doppler frequency shift, and $\tau$ is the time slot duration. Here, $v$ is the HST speed, $\lambda_c =  c / f_c $ is the wavelength, where $f_c$ and $c$ are the carrier frequency and the light speed, respectively. Additionally, $\phi ^{\mathrm{D},k}$ and $\delta ^{\mathrm{D},k}$ represent the azimuth and elevation angle of departure (AoD) from the  BS to user $k$, respectively, $\mathbf{a}_{\mathrm{B}}$ denotes the steering vector at the BS can be given by
\begin{align} \label{eq:absv}
  &\mathbf{a}_{\mathrm{B}}( \phi ^{\mathrm{D},k},\delta ^{\mathrm{D},k} ) = [ 1, \cdots ,  e^{\frac{j2\pi d \left( M_{\mathrm{h}}-1 \right)}{\lambda_c}\sin \phi ^{\mathrm{D},k}\cos \delta ^{\mathrm{D},k} } ] \nonumber  \\
  & \quad \quad  \qquad \qquad \, \otimes  [ 1,\cdots,  e^{\frac{j2\pi d \left( M_{\mathrm{v}}-1 \right)}{\lambda_c}\sin \phi ^{\mathrm{D},k}\sin \delta ^{\mathrm{D},k} } ],
\end{align}
where $d = \lambda_c / 2$ denotes the element spacing.

Moreover, $\widetilde{\mathbf{h}}_{\mathrm{D},k} \in \mathbb{C} ^{M \times 1}$ represents the non-line-of-sight (NLoS) component, where each element modeled as an independent CSCG random vector. The distribution follows $\mathcal{C} \mathcal{N} (\mathbf{0}_{M \times 1},\mathbf{I}_{M \times 1})$ which denotes the small-scale fading.

\subsubsection{BS-RIS Channel}
$\mathbf{G} \in \mathbb{C} ^{N\times M}$ is expressed as
\begin{equation}
  \mathbf{G}=\sqrt{PL_{\mathbf{G}}}\left( \sqrt{\frac{\kappa}{\kappa+1}}\overline{\mathbf{G}}+\sqrt{\frac{1}{\kappa+1}}\widetilde{\mathbf{G}} \right),
\end{equation}
where $PL_{\mathbf{G}}$ is the path loss, which can be given in logarithmic form by 
\begin{equation}
  PL_{\mathbf{G}}=PL_0-10\beta _{\mathbf{G}}\log _{10}( \frac{d_0}{d_{\mathbf{G}}}),
\end{equation}
where $\beta _{\mathbf{G}}$ denotes the path loss exponent, and $d_{\mathbf{G}}$ denotes the distance between the BS and the RIS. The LoS component $\overline{\mathbf{G}} \in \mathbb{C} ^{N\times M}$ is defined as
\begin{equation}
  \overline{\mathbf{G}}=\mathbf{a}_{\mathrm{R}}( \phi ^{\mathrm{R}},\delta ^{\mathrm{R}} ) \mathbf{a}_{\mathrm{B}}^{\mathrm{H}}( \phi ^{\mathrm{B}},\delta ^{\mathrm{B}} ),
\end{equation}
where $\phi ^{\mathrm{R}}$ and $\delta ^{\mathrm{R}}$ are the  azimuth and elevation AoD from the RIS to the BS, respectively. $\phi ^{\mathrm{B}}$ and $\delta ^{\mathrm{B}}$ are the azimuth and elevation angle of arrival (AoA) from the RIS to the BS, respectively. $\mathbf{a}_{\mathrm{R}}$ denotes the steering vector at the RIS, and it is similar to \eqref{eq:absv}. The term $\widetilde{\mathbf{G}} \in \mathbb{C} ^{N\times M}$ is the NLoS component.

\subsubsection{RIS-user Channel}
$\mathbf{h}_{\mathrm{R},k}\in \mathbb{C} ^{N\times 1}$ can be given by 
\begin{equation}
  \mathbf{h}_{\mathrm{R},k}=\sqrt{PL_{\mathrm{R},k}}\left( \sqrt{\frac{\kappa}{\kappa+1}}\overline{\mathbf{h}}_{\mathrm{R},k}+\sqrt{\frac{1}{\kappa+1}}\widetilde{\mathbf{h}}_{\mathrm{R},k} \right),
\end{equation}
where $PL_{\mathrm{R},k} $ is the path loss and can be  expressed in logarithmic form as 
\begin{equation}
  PL_{\mathrm{R},k}=PL_0-10\beta _{\mathrm{R},k}\log _{10}(\frac{d_0}{ d_{\mathrm{R},k}} ),
\end{equation}
where $\beta _{\mathrm{R},k}$ denotes the path loss exponent, and $d_{\mathrm{R},k}$ denotes the distance between the RIS and the $k$th user. Additionally,  $\overline{\mathbf{h}}_{{\mathrm{R},k}} \in \mathbb{C} ^{N\times 1 } $ is the LoS component, defined as
\begin{equation}
  \overline{\mathbf{h}}_{\mathrm{R},k}=e^{j2\pi f_{d2}\tau}\mathbf{a}_{\mathrm{R}}\left( \phi ^{\mathrm{R},k},\delta ^{\mathrm{R},k} \right), 
\end{equation}
where $f_{d_2}=\upsilon \cos \phi ^{\mathrm{R},k}\cos \delta ^{\mathrm{R},k}/\lambda _c$ is the Doppler frequency shift, and  $\widetilde{\mathbf{h}}_{\mathrm{R},k}$ is the NLoS component.

\subsection{Two Channel Uncertainty and CSI Error Models}
In RIS-assisted HST communication systems, two kinds of communication channels are present: \textit{1)} the direct channel  $\mathbf{h}_{\mathrm{D},k}$ and \textit{2)} the cascaded channel $\mathbf{H}_k$. The performance of RIS-assisted HST communication is strongly rely on the precision of the direct CSI at the BS (DCSIB) and the CBRUB. Hence, we first describe two possible scenarios for channel uncertainties and subsequently describe two CSI errors models.

\subsubsection{Partial Channel Uncertainty (PCU)}
In RIS-assisted HST systems, the CBRUB poses more difficulties to estimate compared to the DCSIB, mainly because of the high mobility of HST and the passive characteristic of the RIS. Thus, we assume the perfect DCSIB, and the imperfect CBRUB. The cascaded channel can be formulated as
\begin{equation} \label{eq:cc}
  \mathbf{H}_k=\widehat{\mathbf{H}}_k+\bigtriangleup \mathbf{H}_k,\forall k\in \mathcal{K},
\end{equation}
where $\widehat{\mathbf{H}}_k$ denotes the estimated cascaded CSI which is obtained by the BS, $\bigtriangleup \mathbf{H}_k$ denotes the uncertain CBRUB error.

\subsubsection{Full Channel Uncertainty (FCU)}
In a complex electromagnetic environment, obtaining an accurate DCSIB is also difficult. Hence, we assume that both the DCSIB and the CBRUB are imperfect. Along with the CBRUB error mode in \eqref{eq:cc}, the DCSIB is modeled as
\begin{equation} \label{eq:dc}
  \mathbf{h}_{\mathrm{D},k}=\widehat{\mathbf{h}}_{\mathrm{D},k}+\bigtriangleup \mathbf{h}_{\mathrm{D},k},\forall k\in \mathcal{K},
\end{equation}
where $\widehat{\mathbf{h}}_{\mathrm{D},k}$ represents the estimated DCSIB obtained at the BS, and $\bigtriangleup \mathbf{h}_{\mathrm{D},k}$ is the uncertain DCSIB error.

In this paper, two types of robust beamforming design for RIS-aided HST communication coverage enhancement under different CSI error models is considered.
\subsubsection{Bounded CSI Error (BCSIE) Model }
It specifies the boundaries of the CSI error model \cite{c3}, and given by
\begin{equation} \label{eq:bounded}
  \left\| \bigtriangleup \mathbf{H}_k \right\| _F\le \xi _{\mathrm{H},k},\left\| \bigtriangleup \mathbf{h}_{\mathrm{D},k} \right\| _2\le {\xi _{\mathrm{D}}}_{,k},\forall k\in \mathcal{K},
\end{equation}
where  $\xi _{\mathrm{H},k}$ and $\xi _{\mathrm{D},k}$ represent the radii of the uncertainty areas for the two channels known at the BS.

\subsubsection{Statistical CSI Error (SCSIE) Model }
It models the CSI error through the statistical framework \cite{c4}, and given by
\begin{subequations} \label{eq:statistical}
  \begin{align}
    \mathrm{vec}\left( \bigtriangleup \mathbf{H}_k \right) \sim \mathcal{C} \mathcal{N} (\mathbf{0},\mathbf{\Sigma }_{\mathrm{H},k}),\mathbf{\Sigma }_{\mathrm{H},k}\succeq \mathbf{0},\forall k\in \mathcal{K},  \label{eq:sta1}    \\
    \bigtriangleup \mathbf{h}_{\mathrm{D},k}\sim \mathcal{C} \mathcal{N} (\mathbf{0},\mathbf{\Sigma }_{\mathrm{D},k}),\mathbf{\Sigma }_{\mathrm{D},k}\succeq \mathbf{0},\forall k\in \mathcal{K},  \label{eq:sta2}
  \end{align}
\end{subequations}
where $\mathbf{\Sigma }_{\mathrm{H},k} \in \mathbb{C} ^{NM\times NM} $ and $\mathbf{\Sigma }_{\mathrm{D},k} \in \mathbb{C} ^{M\times M}$ represent covariance matrices of error that are positive semidefinite for the corresponding two channels known at the BS.

\section{Worst-Case Robust Beamforming Design For The BCSIE Model}  \label{wcrbd}
Under the BCSIE model, we investigate the worst-case robust beamforming design. We aim to minimize the BS's transmit power by jointly optimizing  the beamforming matrix $\mathbf{W}$ and reflection beamforming vector $\boldsymbol{\theta }$, while ensuring that the minimum QoS requirements for users under the worst-case conditions and the RIS's unit-modulus constraints are met. To address this (NC) problem, we introduce an alternate optimization (AO) algorithm that utilizes the S-Procedure, SOCP, and penalty CCP techniques \cite{c5}.

\subsection{Scenario 1: PCU}
In this section, our objective is to design the robust beamforming method for the RIS-assisted HST communication converge enhancement under Scenario 1 with the perfect DCSIT and the imperfect CBIUT. Let $\mathcal{I} _{k}^{p}\triangleq \{ \forall \| \bigtriangleup \mathbf{H}_k \| _F\le \xi _{\mathrm{H},k} \}$, which defines the worst-case optimization problem for minimizing the  transmit power for the BS is expressed as
\begin{subequations}
  \begin{align}   \label{eq:Pro1}
    \mathcal{P}_1:\,\, \underset{\mathbf{W},\boldsymbol{\theta }\,\,} {\min}\,\,& \left\| \mathbf{W} \right\| _{F}^{2}     \\
    \mathrm{s}.\mathrm{t}. \,\,\, &  R_k\left( \mathbf{W},\boldsymbol{\theta } \right) \ge R_{\mathrm{th}},\mathcal{I} _{k}^{p},\forall k\in \mathcal{K}, \label{eq:Pro1b} \\
    & \left| \theta _n \right|^2=1,\forall n\in \mathcal{N},  \label{eq:Pro1c}
  \end{align}
\end{subequations}
where $R_{\mathrm{th}}$ denotes the target rate of the $k$th user. The worst-case QoS requirements for users is defined in constraints \eqref{eq:Pro1b}, while constraints \eqref{eq:Pro1c} ensure that the reflection elements at the RIS meet the unit-modulus constraints.

The non-convexity of constraints \eqref{eq:Pro1b} is handled by considering the INs power $\boldsymbol{\iota }=[ \iota _1,\cdots ,\iota _K ] ^{\mathrm{T}}$ as auxiliary variables. As a result, constraints \eqref{eq:Pro1b} is represented as follows:
\begin{align}
 &( \mathbf{h}_{\mathrm{D},k}^{\mathrm{H}}+\boldsymbol{\theta }^{\mathrm{H}}\mathbf{H}_k ) \mathbf{w}_k |^2\ge \iota _k( 2^{R_{\mathrm{th}}}-1 ) ,\mathcal{I} _{k}^{p},\forall k\in \mathcal{K}, \label{eq:INs1}   \\
  &\| ( \mathbf{h}_{\mathrm{D},k}^{\mathrm{H}}+\boldsymbol{\theta }^{\mathrm{H}}\mathbf{H}_k ) \mathbf{W}_{-k}\| _{2}^{2}+\sigma _{k}^{2}\le \iota _k,\mathcal{I} _{k}^{p},\forall k\in \mathcal{K}.   \label{eq:INs2}
\end{align}
Constraints \eqref{eq:INs1} and \eqref{eq:INs2} define the limits for the power of the worst-case useful signal and INs, respectively.

To address the NC semi-infinite inequality constraints \eqref{eq:INs1}, we begin by approximating the NC terms, followed by applying the S-Procedure to handle the semi-infinite inequalities. The useful signal power in \eqref{eq:INs1} can be converted to a linear approximation by the following lemma.
\begin{lemma} \label{le1}
By inserting $\mathbf{H}_k=\widehat{\mathbf{H}}_k+\bigtriangleup \mathbf{H}_k$ into the useful signal power expression \eqref{eq:INs1}, and treating $\mathbf{w}_{k}^{\left( n \right)}$ and $\boldsymbol{\theta }_{k}^{\left( n \right)}$ as the optimal solutions acquired at the $n$-th iteration, the lower bound at $( \mathbf{w}_{k}^{\left( n \right)},\boldsymbol{\theta }_{k}^{\left( n \right)} )$ provides a linear approximation of the expression $| [ \mathbf{h}_{\mathrm{D},k}^{\mathrm{H}}+\boldsymbol{\theta }^{\mathrm{H}}( \widehat{\mathbf{H}}_k+\bigtriangleup \mathbf{H}_k ) ] \mathbf{w}_k |^2$, as follows
 \begin{equation}    \label{eq:le1}
    \mathrm{vec}^\mathrm{T}(\bigtriangleup \mathbf{H}_k)\mathbf{A}_k\mathrm{vec(}\bigtriangleup \mathbf{H}_{k}^{*})+2\mathfrak{R} \left\{ \mathbf{e}_{k}^{\mathrm{T}}\mathrm{vec(}\bigtriangleup \mathbf{H}_{k}^{*}) \right\} +e_k,
  \end{equation}
  where
  \begin{align*}
    \mathbf{A}_k & =   \mathbf{w}_k\mathbf{w}_{k}^{\left(n\right),\mathrm{H}}\otimes \boldsymbol{\theta }^*\boldsymbol{\theta }^{\left(n\right),\mathrm{T}}+\mathbf{w}_{k}^{\left(n\right)}\mathbf{w}_{k}^{\mathrm{H}}\otimes \boldsymbol{\theta }^{\left(n\right),*}\boldsymbol{\theta }^{\mathrm{T}}   \\
    &\thinspace\thinspace\thinspace\thinspace\thinspace\thinspace\thinspace -(\mathbf{w}_{k}^{\left(n\right)}\mathbf{w}_{k}^{\left(n\right),\mathrm{H}}\otimes \boldsymbol{\theta }^{\left(n\right),*}\boldsymbol{\theta }^{\left(n\right),\mathrm{T}}), \\
    \mathbf{e}_k&=  \mathrm{vec(}\boldsymbol{\theta }( \mathbf{h}_{\mathrm{D},k}^{{\mathrm{H}}}+\boldsymbol{\theta }^{\left(n\right),\mathrm{H}}\widehat{\mathbf{H}}_k ) \mathbf{w}_{k}^{\left(n\right)}\mathbf{w}_{k}^{\mathrm{H}}) \\
    &\thinspace\thinspace\thinspace\thinspace\thinspace\thinspace\thinspace+\mathrm{vec(}\boldsymbol{\theta }^{\left(n\right)}( \mathbf{h}_{\mathrm{D},k}^{\mathrm{H}}+\boldsymbol{\theta }^\mathrm{H}\widehat{\mathbf{H}}_k ) \mathbf{w}_k\mathbf{w}_{k}^{\left(n\right),\mathrm{H}}) \\
    &\thinspace\thinspace\thinspace\thinspace\thinspace\thinspace\thinspace-\mathrm{vec(}\boldsymbol{\theta }^{\left(n\right)}( \mathbf{h}_{\mathrm{D},k}^{\mathrm{H}}+\boldsymbol{\theta }^{\left(n\right),\mathrm{H}}\widehat{\mathbf{H}}_k ) \mathbf{w}_{k}^{\left(n\right)}\mathbf{w}_{k}^{\left(n\right),\mathrm{H}}), \\
    e_k & =   2\mathfrak{R} \{ ( \mathbf{h}_{\mathrm{D},k}^{\mathrm{H}}+\boldsymbol{\theta }^{\left(n\right),\mathrm{H}}\widehat{\mathbf{H}}_k  \mathbf{w}_{k}^{\left(n\right)}\mathbf{w}_{k}^{\mathrm{H}}( \mathbf{h}_{\mathrm{D},k}+\widehat{\mathbf{H}}_{k}^{\mathrm{H}}\boldsymbol{\theta } ) \} \\
    &\thinspace\thinspace\thinspace\thinspace\thinspace\thinspace\thinspace-( \mathbf{h}_{\mathrm{D},k}^{\mathrm{H}}+\boldsymbol{\theta }^{\left(n\right),\mathrm{H}}\widehat{\mathbf{H}}_k ) \mathbf{w}_{k}^{\left(n\right)}\mathbf{w}_{k}^{\left( n \right) ,\mathrm{H}}( \mathbf{h}_{\mathrm{D},k}+\widehat{\mathbf{H}}_{k}^{\mathrm{H}}\boldsymbol{\theta }^{\left(n\right)} ).
  \end{align*}
\end{lemma}

\begin{proof}
  Kindly consult Appendix \ref{appa}.
\end{proof}

By using the linear approximation of the useful signal power in \eqref{eq:le1} instead of the original in \eqref{eq:INs1}, we obtain a reformulation of the constraints in \eqref{eq:INs1} as
\begin{align} \label{eq:linearapp}
  &\mathrm{vec}^{\mathrm{T}}(\bigtriangleup \mathbf{H}_k)\mathbf{A}_k\mathrm{vec(}\bigtriangleup \mathbf{H}_{k}^{*})+2\mathfrak{R} \{ \mathbf{e}_{k}^{\mathrm{T}}\mathrm{vec(}\bigtriangleup \mathbf{H}_{k}^{*}) \} +e_k \nonumber \\
  &\ge \iota _k(2^{R_{\mathrm{th}}}-1),\mathcal{I} _{k}^{p},\forall k\in \mathcal{K}.
\end{align}

Nevertheless, the constraints in \eqref{eq:linearapp} are still challenging to solve because they are semi-infinite inequalities. Thus, we have the following lemma to address this problem.
\begin{lemma} \label{S-procedure} (General S-Procedure \cite{c6})
Considering the quadratic functions associated with the variable$\mathbf{x}\in\mathbb{C}^{n\times1}$:
\begin{equation*}
  f_i(\mathbf{x})=\mathbf{x}^{\mathrm{H}}\mathbf{E}_i\mathbf{x}+2\mathfrak{R} \left\{ \mathbf{e}_{i}^{\mathrm{H}}\mathbf{x} \right\} +g_i,i=0,...,P,
\end{equation*}
where $\mathbf{E}_i=\mathbf{E}_{i}^{\mathrm{H}}$. The condition $\left\{ f_i\left( \mathbf{x} \right) \ge 0 \right\} _{i=1}^{P}\Rightarrow f_0\left( \mathbf{x} \right) \ge 0$ holds if and only if there exist $\forall i,\varrho _i\ge 0$ such that
\begin{equation*}
  \left. \left[ \begin{matrix}
    \mathbf{E}_0&		\mathbf{e}_0\\
    \mathbf{e}_{0}^{\mathrm{H}}&		g_0\\
  \end{matrix} \right. \right] -\sum_{i=1}^P{\varrho _i\left[ \begin{matrix}
    \mathbf{E}_i&		\mathbf{e}_i\\
    \mathbf{e}_{i}^{\mathrm{H}}&		g_i\\
  \end{matrix} \right]}\succeq \mathbf{0}.
\end{equation*}
\end{lemma}

Using Lemma \ref{S-procedure} to address the CSI uncertainty in the constraints \eqref{eq:linearapp} allows for the transformation into equivalent linear matrix inequalities (LMIs). In particularly, each user's constraint in \eqref{eq:linearapp} can be rewritten by appropriately choosing the parameters as specified in Lemma \ref{S-procedure} as
\begin{align*}
  &P=1,\thinspace\thinspace\mathbf{E}_0=\mathbf{A}_k,\thinspace\thinspace\mathbf{w}_0=\mathbf{e}_{k},\thinspace\thinspace g_0=e_k-\iota _k\left( 2^{R_{\mathrm{th}}}-1 \right),
   \\
  &\mathbf{x}=\mathrm{vec}\left( \bigtriangleup \mathbf{H}_{k}^{*} \right),\thinspace\thinspace\thinspace\thinspace\mathbf{E}_1=-\mathbf{I},\thinspace\thinspace\thinspace\thinspace g_1=\xi _{\mathrm{H},k}^{2}.
\end{align*}

Subsequently, constraint \eqref{eq:linearapp} is transformed into the following equivalent form of LMIs;
\begin{equation}
  \left[ \begin{matrix}
    \vartheta _{\mathrm{H},k}\mathbf{I}_{NM}+\mathbf{A}_k&		\mathbf{e}_k\\
    \mathbf{e}_{k}^{\mathrm{T}}&		C_{k}^{p}\\
  \end{matrix} \right] \succeq \mathbf{0},\forall k\in \mathcal{K},
  \label{eq:LMI-S-partial}
\end{equation}
where $\boldsymbol{\vartheta }_{\mathrm{H}}=\left[ \vartheta _{\mathrm{H},1},\cdots ,\vartheta _{\mathrm{H},K} \right] ^{\mathrm{T}}\ge 0$ represent slack variables and $C_{k}^{p}=e_k-\iota _k( 2^{R_{\mathrm{th}}}-1) -\vartheta _{\mathrm{H},k}\xi _{\mathrm{H},k}^{2}$.

In the following, we analyze the uncertainty in $\left\{ \bigtriangleup \mathbf{H}_k \right\} _{\forall k\in \mathcal{K}}$ from \eqref{eq:INs2}. To begin with, we apply Schur's complement Lemma \cite{c7} to reformulate the INs power inequalities in \eqref{eq:INs2} into their corresponding matrix inequalities, as given below
\begin{equation}
  \left. \left[ \begin{matrix}
    \iota _k-\sigma _{k}^{2}&		\mathbf{r}_{k}^{\mathrm{H}}\\
    \mathbf{r}_k&		\mathbf{I}\\
  \end{matrix} \right. \right] \succeq \mathbf{0},\forall k\in \mathcal{K},
  \label{eq:IN-LMI-1-1}
\end{equation}
where $\mathbf{r}_k=( ( \mathbf{h}_{\mathrm{D},k}^{\mathrm{H}}+\boldsymbol{\theta }^{\mathrm{H}}\mathbf{H}_k ) \mathbf{W}_{-k} ) ^{\mathrm{H}}$. By using $\mathbf{H}_k=\widehat{\mathbf{H}}_k+\bigtriangleup \mathbf{H}_k$, \eqref{eq:IN-LMI-1-1} can be reformulated as
\begin{align}
  \left[ \begin{matrix}
    \iota _k-\sigma _{k}^{2}&		\widehat{\mathbf{r}}_{k}^{\mathrm{H}}\\
    \widehat{\mathbf{r}}_k&		\mathbf{I}\\
  \end{matrix} \right] \succeq -\left[ \begin{array}{c}
    \mathbf{0}\\
    \mathbf{W}_{-k}^{\mathrm{H}}\\
  \end{array} \right] \bigtriangleup \mathbf{H}_{k}^{\mathrm{H}}\left[ \begin{matrix}
    \boldsymbol{\theta }&		\mathbf{0} \nonumber    \\
  \end{matrix} \right]
  \\
  \,\,                -\left[ \begin{array}{c}
    \boldsymbol{\theta }^\mathrm{H} \\
    \mathbf{0}  \\
  \end{array} \right] \bigtriangleup \mathbf{H}_k\left[ \mathbf{0}\,\,\mathbf{W}_{-k} \right] ,\forall k\in \mathcal{K},
  \label{eq:LMI-IN-2}
\end{align}
where $\widehat{\mathbf{r}}_k=( ( \mathbf{h}_{\mathrm{D},k}^{\mathrm{H}}+\boldsymbol{\theta }^{\mathrm{H}}\widehat{\mathbf{H}}_k ) \mathbf{W}_{-k} ) ^{\mathrm{H}}$.

In tackling \eqref{eq:LMI-IN-2}, the following lemma is presented.
\begin{lemma} \label{sign-definiteness}(General sign-definiteness \cite{c8})
For a defined set of matrices, such that $\mathbf{E}=\mathbf{E}^{\mathrm{H}}$, and $\left\{ \mathbf{Y}_i,\mathbf{Z}_i \right\} _{i}^{P}$, the subsequent LMI holds
\begin{equation*}
  \mathbf{E}\succeq \sum_{i=1}^P{\left( \mathbf{Y}_{i}^{\mathrm{H}}\mathbf{X}_i\mathbf{Z}_i+\mathbf{Z}_{i}^{\mathrm{H}}\mathbf{X}_{i}^{\mathrm{H}}\mathbf{Y}_i \right) ,}\forall i,||\mathbf{X}_i||_F\le \zeta _i,
\end{equation*}
if and only if there exist real numbers $\forall i,\varsigma _i\ge 0$ such that
\begin{equation*}
  \left[ \begin{matrix}
    \mathbf{E}-\sum_{i=1}^P{\varsigma _i\mathbf{Z}_{i}^{\mathrm{H}}\mathbf{Z}_i}&		-\zeta _1\mathbf{Y}_{1}^{\mathrm{H}}&		\cdots&		-\zeta _P\mathbf{Y}_{P}^{\mathrm{H}}\\
    -\zeta _1\mathbf{Y}_1&		\varsigma _1\mathbf{I}&		\cdots&		0\\
    \vdots&		\vdots&		\ddots&		\vdots\\
    -\zeta _P\mathbf{Y}_P&		0&		\cdots&		\varsigma _P\mathbf{I}\\
  \end{matrix} \right] \succeq \mathbf{0}.
\end{equation*}
\end{lemma}

By utilizing Lemma \ref{S-procedure}, it can be proved, and a comprehensive proof is available in \cite{c9}.

For the utilization of Lemma \ref{sign-definiteness}, the following parameters are selected (It is worth mentioning that the notation $i $ used in Lemma \ref{sign-definiteness} is omitted because $P =1$.) for each of the constraints in \eqref{eq:LMI-IN-2} as
\begin{equation*}
  \mathbf{E}=\left[ \begin{matrix}
    \iota _k-\sigma _{k}^{2}&		\hat{\mathbf{r}}_{k}^{\mathrm{H}}\\
    \hat{\mathbf{r}}_k&		\mathbf{I}\\
  \end{matrix} \right] ,\mathbf{Y}=-\left[ \mathbf{0}\,\,\mathbf{W}_{-k}^{\mathrm{H}} \right] ,\mathbf{Z}=\left[ \boldsymbol{\theta }\,\,\mathbf{0} \right] ,\mathbf{X}=\bigtriangleup \mathbf{H}_k.
\end{equation*}

Then, by introducing slack variables $\boldsymbol{\mu }_{\mathrm{H}}=\left[ \mu _{\mathrm{H},1},\cdots ,\mu _{\mathrm{H},K} \right] ^{\mathrm{T}}\ge 0$, using Lemma \ref{sign-definiteness}, constraints \eqref{eq:INs2} is converted into the following equivalent LMI,
\begin{equation}
  \left[ \begin{matrix}
    \iota _k-\sigma _{k}^{2}-\mu _{\mathrm{H},k}N&		\widehat{\mathbf{r}}_{k}^{\mathrm{H}}&		0_{1\times M}\\
    \widehat{\mathbf{r}}_k&		\mathbf{I}_{(K-1)}&		\xi _{\mathrm{H},k}\mathbf{W}_{-k}^{\mathrm{H}}\\
    0_{N\times 1}&		\xi _{\mathrm{H},k}\mathbf{W}_{-k}&		\mu _{\mathrm{H},k}\mathbf{I}_M\\
  \end{matrix} \right] \succeq \mathbf{0},\forall k\in \mathcal{K},
  \label{eq:LMI-IN-1-1}
\end{equation}

Therefore, the problem $\mathcal{P}_1$ is approximately rewritten as
\begin{subequations}   \label{eq:Pro2}
  \begin{align}
    \mathcal{P}_2:\,\,  \underset{\mathbf{W},\boldsymbol{\theta },\boldsymbol{\iota },\boldsymbol{\vartheta }_{\mathrm{H}},\boldsymbol{\mu }_{\mathrm{H}}\,\,}{\min}\,\, & \left\| \mathbf{W} \right\| _{F}^{2}     \\
    \mathrm{s}.\mathrm{t}. \,\,\,\,\,\,\,\,\,\,\,\, &  \eqref{eq:LMI-S-partial}, \eqref{eq:LMI-IN-1-1},\eqref{eq:Pro1c},  \label{eq:Pro2b} \\
    & \boldsymbol{\vartheta }_{\mathrm{H}}\ge 0,\boldsymbol{\mu }_{\mathrm{H}}\ge 0.  \label{eq:Pro2c}
  \end{align}
\end{subequations}
$\mathcal{P}_2$ remains NC and challenging to solve due to the coupling between  $\mathbf{W}$ and $\boldsymbol{\theta }$ in $\mathbf{A}_k$, $\mathbf{a}_k$, and $\widehat{\mathbf{r}}_k$. To overcome this, we decouple it into two tractable subproblems by using the AO method, which can be solved iteratively to obtain optimal solutions. Specifically, we begin by minimizing the transmit power with the reflection beamforming $\boldsymbol{\theta}$ fixed, which transforms the problem into a convex one in terms of $\mathbf{W}$. The convex problem can be tackled with the help of the CVX tool. After $\mathbf{W}$ is fixed, we use the penalty CCP method to address the NC problem of $\boldsymbol{\theta}$. For a given $\boldsymbol{\theta}$, the corresponding subproblem for $\mathbf{W}$ is expressed as
\begin{subequations}
\begin{align} \label{eq:Pro3}
  \mathcal{P}_3:\,\, \mathbf{W}^{\left( n+1 \right)}=\mathrm{arg}\underset{\mathbf{W},\boldsymbol{\iota },\boldsymbol{\vartheta }_{\mathrm{H}},\boldsymbol{\mu }_{\mathrm{H}}\,\,}{\min}\,\, & \left\| \mathbf{W} \right\| _{F}^{2} \\
  \mathrm{s}.\mathrm{t}. \, \,\,\,\,\,\,\,\,\,  & \eqref{eq:LMI-S-partial}, \eqref{eq:LMI-IN-1-1},  \eqref{eq:Pro2c},  \label{eq:Pro3b}
\end{align}
\end{subequations}
where $\mathbf{W}^{( n+1 )}$ denotes the solution achieved as optimal in the $(n+1)$-th iteration. $\mathcal{P}_3$ represents a semidefinite program (SDP), which can be addressed via the CVX tool.

Next, with $\mathbf{W}$ fixed, the subproblem related to $\boldsymbol{\theta}$ becomes a feasibility-check problem. To improve the solution that is converged during the optimization of $\boldsymbol{\theta}$, as outlined in \cite{c11,c12}, the inequalities for the useful signal power in \eqref{eq:INs1} are adjusted by incorporating slack variables $\boldsymbol{\alpha }=\left[ \alpha _1,\dots ,\alpha _K \right] ^{\mathrm{T}}\ge 0$ and can be reformulated by
\begin{equation}
  | ( \mathbf{h}_{\mathrm{D},k}^{\mathrm{H}}+\boldsymbol{\theta }^{\mathrm{H}}\mathbf{H}_k ) \mathbf{w}_k |^2\ge \iota _k( 2^{R_{\mathrm{th}}}-1 ) +\alpha _k,\forall k\in \mathcal{K}.
  \label{eq:SINR-slac}
\end{equation}
Therefore, the LMIs \eqref{eq:LMI-S-partial} can be rewritten as
\begin{equation}
  \left[ \begin{matrix}
    \vartheta _{\mathrm{H},k}\mathbf{I}_{NM}+\mathbf{A}_k&		\mathbf{a}_k\\
    \mathbf{a}_{k}^{\mathrm{T}}&		C_{k}^{p}-\alpha _k\\
  \end{matrix} \right] \succeq \mathbf{0},\forall k\in \mathcal{K},
  \label{eq:LMI-S-partial-1}
\end{equation}
Furthermore, we notice that only the $K \times K$ submatrix in the upper left corner of \eqref{eq:LMI-IN-1-1} is affected by $\boldsymbol{\theta}$. As a result, the dimension of the LMIs in \eqref{eq:LMI-IN-1-1} is reduced from  $(K+M)\times(K+M)$ to $K\times K$ as shown below:
\begin{equation}
  \left[ \begin{matrix}
    \iota _k-\sigma _{k}^{2}-\mu _{\mathrm{H},k}N&		\widehat{\mathbf{r}}_{k}^{\mathrm{H}}\\
    \widehat{\mathbf{r}}_k&		\mathbf{I}_{(K-1)}\\
  \end{matrix} \right] \succeq \mathbf{0},\forall k\in \mathcal{K}.
  \label{eq:LMI-IN-partial-1}
\end{equation}
By combining \eqref{eq:LMI-S-partial-1}  and \eqref{eq:LMI-IN-partial-1}, the subproblem of $\boldsymbol{\theta}$ can be expressed as
\begin{subequations}
  \begin{align}   \label{eq:Pro4}
    \mathcal{P}_4:\,\,  \max_{\boldsymbol{\alpha },\boldsymbol{\theta },\boldsymbol{\iota },\boldsymbol{\vartheta }_{\mathrm{H}},\boldsymbol{\mu }_{\mathrm{H}}} & \sum_{k=1}^K{\alpha _k}     \\
    \mathrm{s}.\mathrm{t}. \,\,\,\,\,\,\,\,\, &  \eqref{eq:LMI-S-partial-1}, \eqref{eq:LMI-IN-partial-1}, \eqref{eq:Pro1c}, \eqref{eq:Pro2c},  \label{eq:Pro4b} \\
    & \boldsymbol{\alpha } \ge 0.  \label{eq:Pro4c}
  \end{align}
\end{subequations}
It should be emphasized that solving  $\mathcal{P}_4$ can lead to a lower objective value than  $\mathcal{P}_3$. Further details on this can be found in \cite{c12}.

We observe that the non-convexity of  $\mathcal{P}_4$ is caused by unit-modulus constraints in \eqref{eq:Pro1c}. Hence, we utilize the penalty CCP method \cite{c5} to transform these constraints into convex forms. Following the method in \cite{c11}, \eqref{eq:Pro1c} is  first expressed as $1\le | \theta _n |^2\le 1,\forall n\in \mathcal{N} $, and the resulting constraints, which involve NC components, are linearized as $| \theta _{n}^{[ t ]} |^2-2\mathfrak{R} ( \theta _{n}^{*}\theta _{n}^{[ t ]} ) \le -1,\forall n\in \mathcal{N} $, with $\theta _{n}^{[ t ]}$ held fixed. This allows for the reformulation of  $\mathcal{P}_4$ as
\begin{subequations}
  \begin{align}   \label{eq:Pro5}
    \mathcal{P}_5:
    &\max_{\boldsymbol{\alpha },\boldsymbol{\theta },\mathbf{b},\boldsymbol{\iota },\boldsymbol{\vartheta }_{\mathrm{H}},\boldsymbol{\mu }_{\mathrm{H}}} \sum_{k=1}^K{\alpha _k-\mathrm{\nu}^{\left[ \mathrm{t} \right]}\sum_{n=1}^{2N}{b_n}}    \\
    & \ \ \ \ \ \mathrm{s}.\mathrm{t}. \thinspace\thinspace \eqref{eq:LMI-S-partial-1}, \eqref{eq:LMI-IN-partial-1}, \eqref{eq:Pro2c}, \eqref{eq:Pro4c},  \label{eq:Pro5b} \\
    & \ \ \ \ \ \ \ \ \ \   | \theta _{n}^{\left[ t \right]} |^2-2\mathfrak{R} ( \theta _{n}^{*}\theta _{n}^{\left[ t \right]} ) \le b_n-1,\forall n\in \mathcal{N},  \label{eq:Pro5c} \\
    & \ \ \ \ \ \ \ \ \ \  | \theta _n |^2\le 1+b_{N+n},\forall n\in \mathcal{N},  \label{eq:Pro5d} \\
    & \ \ \ \ \ \ \ \ \ \  \mathbf{b}\ge 0.  \label{eq:Pro5e}
  \end{align}
\end{subequations}
where $\mathbf{b}=[ b_1,\cdots,b_{2N} ] ^{\mathrm{T}}$ represents slack variables introduced to convert the unit-modulus constraints into equivalent linear constraints. The term $\left\| \mathbf{b} \right\| _1$ functions as a penalty within the objective function.  Furthermore, the regularization parameter
$\mathrm{\nu}^{\left[ \mathrm{t} \right]}$ scales
$\| \mathbf{b} \| _1$ to control feasibility of the constraint.

$\mathcal{P}_5$ is an SDP, and its can be derived using the CVX tool to obtain the suboptimal solutions. Furthermore, the procedure for the proposed penalty CCP algorithm related to reflecting beamforming is detailed in Algorithm~\ref{PCCP}. There some points are worth our attention: 1) The effectiveness of constraints \eqref{eq:Pro1c} in the initial $\mathcal{P}_4$ holds when the condition $\left\| \mathbf{b} \right\| _1\le \mathcal{Y} $ is met, provided that $\mathcal{Y}$ remains sufficiently small; 2) To mitigate numerical problems, we introduce a maximum value $\nu _{\max}$. Specifically, as $\nu _{\max}$ increases, once the iteration satisfies the stopping criterion $\| \theta ^{[ t ]}-\theta ^{[ t-1 ]} \| _1\le \epsilon _{\mathrm{stop}} $, it may be impossible to find a solution that satisfies $\| \mathbf{b} \| _1\le \mathcal{Y}$; 3) The convergence of the Algorithm~\ref{PCCP} is governed by the stopping criterion $\| \theta ^{[ t ]}-\theta ^{[ t-1 ]} | _1\le \epsilon _{\mathrm{stop}} $;
5)  As pointed out in \cite{c5}, a solution that is feasible for solution for $\mathcal{P}_5$ might not fulfill the feasibility requirements for $\mathcal{P}_4$. Thus, the feasibility of $\mathcal{P}_4$ is ensured by restricting the number of iterations $T_{\max}$. If this limit is reached, the iteration process is restarted with a new initial point. Subsequently, $\mathcal{P}_2$ is tackled by iteratively solving  $\mathcal{P}_3$ and $\mathcal{P}_4$ within the framework of the AO algorithm. It is important to note that the fixed point $\boldsymbol{\theta}^{[ t ]}$ in constraint \eqref{eq:Pro5c} is modified at every iteration of Algorithm~\ref{PCCP}, which corresponds to $\mathrm{\nu}^{[ \mathrm{t} ]}$.
\begin{algorithm}[!t]
  \caption{Penalty CCP Algorithm for Optimizing $\boldsymbol{\theta}$}
  \label{PCCP}
  \begin{algorithmic}[1] \REQUIRE Initialize
  $\boldsymbol{\theta}^{\left[ 0\right] }$, $t=0$, $\gamma^{\left[ 0\right] }>1$.
  \REPEAT
  \IF {$t<T_{\max}$ }
  \STATE Update $\boldsymbol{\theta}^{[t+1]}$ according to the conditions of $\mathcal{P}_5$;
  \STATE $\nu^{[t+1]}=\min \{ \gamma\nu^{\left[ t\right] }, \nu _{\max}\}$;
  \STATE $t=t+1$;
  \ELSE
  \STATE Begin with a new random $\boldsymbol{\theta}^{\left[ 0\right] }$, reset $\gamma^{\left[ 0\right] }$ to be greater than $1$, and set  the iteration index $t=0$.
  \ENDIF
  \UNTIL $||\mathbf{b}||_{1}\leq\chi$ and $||\boldsymbol{\theta}^{\left[ t\right] }-\boldsymbol{\theta}^{\left[ t-1\right] }||_{1}\leq  \epsilon _{\mathrm{stop}} $.
  \STATE Output $\boldsymbol{\theta}^{ \left( n+1\right)}=\boldsymbol{\theta}^{\left[ t\right] }$.
  \end{algorithmic}
\end{algorithm}

\subsection{Scenario 2: FCU}
This section discusses the design of robust beamforming, assuming imperfections in both the DCSIB and CBRUB. By incorporating the FCU as defined in \eqref{eq:cc} and \eqref{eq:dc}, and defining $\mathcal{I} _{k}^{f}\triangleq \left\{ \forall \left\| \bigtriangleup \mathbf{h}_{\mathrm{D},k} \right\| _2\le {\xi _{\mathrm{D}}}_{,k},\forall \left\| \bigtriangleup \mathbf{H}_k \right\| _F\le \xi _{\mathrm{H},k} \right\}$, the constraints in \eqref{eq:Pro1b} can be extended as follows:
\begin{equation} \label{eq:full-error}
  R_k\left( \mathbf{W},\boldsymbol{\theta } \right) \ge R_{\mathrm{th}},\mathcal{I} _{k}^{f},\forall k\in \mathcal{K},
\end{equation}
which is the equivalent to
\begin{align}
  &| ( \mathbf{h}_{\mathrm{D},k}^{\mathrm{H}}+\boldsymbol{\theta }^{\mathrm{H}}\mathbf{H}_k ) \mathbf{w}_k |^2\ge \iota _k\left( 2^{R_{\mathrm{th}}}-1 \right) ,\mathcal{I} _{k}^{f},\forall k\in \mathcal{K}, \label{eq:full-SINR-signal} \\
  &\|( \mathbf{h}_{\mathrm{D},k}^{\mathrm{H}}+\boldsymbol{\theta }^{\mathrm{H}}\mathbf{H}_k ) \mathbf{W}_{-k} \| _{2}^{2}+\sigma _{k}^{2}\le \iota _k,\mathcal{I} _{k}^{f},\forall k\in \mathcal{K}. \label{eq:full-SINR-interference}
\end{align}

The semi-infinite inequality constraints that are NC mentioned above can be solved in a manner similar to Scenario 1. Particularly, the following lemma provides the linear approximation of the useful signal power, as outlined in \eqref{eq:full-SINR-signal}.
\begin{lemma} \label{lemma-full}
  Let $\mathbf{w}_{k}^{\left( n \right)}$ and $\boldsymbol{\theta }_{k}^{\left( n \right)}$ represent the optimal solutions obtained during the $n$-th iteration. By substituting $\mathbf{h}_{\mathrm{D},k}=\widehat{\mathbf{h}}_{\mathrm{D},k}+\bigtriangleup \mathbf{h}_{\mathrm{D},k}$ and $\mathbf{H}_k=\widehat{\mathbf{H}}_k+\bigtriangleup \mathbf{H}_k$ into the formulation of the useful signal power as shown in \eqref{eq:full-SINR-signal}, the resulting term $| [ ( \widehat{\mathbf{h}}_{\mathrm{D},k}+\bigtriangleup \mathbf{h}_{\mathrm{D},k}) +\boldsymbol{\theta }^{\mathrm{H}}( \widehat{\mathbf{H}}_k+\bigtriangleup \mathbf{H}_k ) ] \mathbf{w}_k |^2$ is bounded below linearly at the point $( \mathbf{w}_{k}^{\left( n \right)},\boldsymbol{\theta }_{k}^{\left( n \right)} ) $, as shown below
  \begin{equation}  \label{eq:lemma-full}
    \mathbf{x}_{k}^{\mathrm{H}}\widetilde{\mathbf{A}}_k\mathbf{x}_k+2\mathfrak{R} \left\{ \widetilde{\mathbf{e}}_{k}^{\mathrm{H}}\mathbf{x}_k \right\} +\widetilde{e}_k,
  \end{equation}
  where
  \begin{align*}
    \widetilde{\mathbf{A}}_k & =\mathbf{J}_k+\mathbf{J}_{k}^{\mathrm{H}}-\mathbf{L}_k,  \\
    \mathbf{J}_k & =\left[ \begin{array}{c}
    \mathbf{w}_{k}^{\left( n \right)}\\
    {\mathbf{w}}_{k}^{\left( n \right)}\otimes \boldsymbol{\theta }^{\left( n \right),*}\\
    \end{array} \right] \left[ \begin{matrix}
    \mathbf{w}_{k}^{\mathrm{H}}&\mathbf{w}_{k}^{\mathrm{H}}\otimes \boldsymbol{\theta }^{\mathrm{T}}
    \end{matrix} \right], \\
    \mathbf{L}_k & =\left[ \begin{array}{c}
    \mathbf{w}_{k}^{\left( n \right)}\\
    \mathbf{w}_{k}^{\left( n \right)}\otimes \boldsymbol{\theta }^{\left( n \right),*}\\
    \end{array} \right] \left[ \begin{matrix}
    \mathbf{w}_{k}^{\left( n \right),\mathrm{H}}&		\mathbf{w}_{k}^{\left( n \right),\mathrm{H}}\otimes \boldsymbol{\theta }^{\left( n \right),\mathrm{T}}\\
    \end{matrix} \right], \\
    \widetilde{\mathbf{e}}_k&=\mathbf{j}_{1,k}+\mathbf{j}_{2,k}-\mathbf{l}_k, \\
    \mathbf{j}_{1,k}&=\left[ \begin{array}{c}
    \mathbf{w}_k\mathbf{w}_{k}^{\left(n\right),\mathrm{H}}( \widehat{\mathbf{h}}_{\mathrm{D},k}+\widehat{\mathbf{H}}_{k}^{\mathrm{H}}\boldsymbol{\theta }^{\left(n\right)} )\\
    \mathrm{vec}^*(\boldsymbol{\theta }( \widehat{\mathbf{h}}_{\mathrm{D},k}+\boldsymbol{\theta }^{\left(n\right),\mathrm{H}}\widehat{\mathbf{H}}_k) \mathbf{w}_{k}^{\left(n\right)}\mathbf{w}_{k}^{\mathrm{H}})\\
    \end{array} \right], \\
    \mathbf{j}_{2,k}&=\left[ \begin{array}{c}
    \mathbf{w}_{k}^{\left(n\right)}\mathbf{w}_{k}^{\mathrm{H}}( \widehat{\mathbf{h}}_{\mathrm{D},k}+\widehat{\mathbf{H}}_{k}^{\mathrm{H}}\boldsymbol{\theta } )\\
    \mathrm{vec}^*(\boldsymbol{\theta }^{\left(n\right)}( \widehat{\mathbf{h}}_{\mathrm{D},k}^{\mathrm{H}}+\boldsymbol{\theta }^{\mathrm{H}}\widehat{\mathbf{H}}_k ) \mathbf{w}_k\mathbf{w}_{k}^{\left(n\right),\mathrm{H}})\\
    \end{array} \right], \\
    \mathbf{l}_k&=\left[ \begin{array}{c}
    \mathbf{w}_{k}^{\left(n\right)}\mathbf{w}_{k}^{\left(n\right),\mathrm{H}}( \widehat{\mathbf{h}}_{\mathrm{D},k}+\widehat{\mathbf{H}}_{k}^{\mathrm{H}}\boldsymbol{\theta }^{\left(n\right)} )\\
    \mathrm{vec}^*(\boldsymbol{\theta }^{\left(n\right)}( \widehat{\mathbf{h}}_{\mathrm{D},k}^{\mathrm{H}}+\boldsymbol{\theta }^{\left(n\right),\mathrm{H}}\widehat{\mathbf{H}}_k ) \mathbf{w}_{k}^{\left(n\right)}\mathbf{w}_{k}^{\left(n\right),\mathrm{H}})\\
    \end{array} \right], \\
    \widetilde{e}_k&=2\mathfrak{R} \left( j_k \right) -l_k, \\
    j_k&=( \widehat{\mathbf{h}}_{\mathrm{D},k}^{\mathrm{H}}+\boldsymbol{\theta }^{\left(n\right),\mathrm{H}}\widehat{\mathbf{H}}_k ) \mathbf{w}_{k}^{\left(n\right)}\mathbf{w}_{k}^{\mathrm{H}}( \widehat{\mathbf{h}}_{\mathrm{D},k}+\widehat{\mathbf{H}}_{k}^{\mathrm{H}}\boldsymbol{\theta } ), \\
    l_k&=( \widehat{\mathbf{h}}_{\mathrm{D},k}^{\mathrm{H}}+\boldsymbol{\theta }^{\left(n\right),\mathrm{H}}\widehat{\mathbf{H}} ) \mathbf{w}_{k}^{\left(n\right)}\mathbf{w}_{k}^{\mathrm{H}}( \widehat{\mathbf{h}}_{\mathrm{D},k}+\widehat{\mathbf{H}}_{k}^{\mathrm{H}}\boldsymbol{\theta }^{\left(n\right)} ),\\
    \mathbf{x}_k&=\left[ \begin{matrix}
    \bigtriangleup \widehat{\mathbf{h}}_{\mathrm{D},k}^{\mathrm{H}}&		\mathrm{vec}^{\mathrm{H}}( \bigtriangleup \widehat{\mathbf{H}}_{k}^{*} )\\
    \end{matrix} \right] ^{\mathrm{H}}.
  \end{align*}
\end{lemma}

\begin{proof}
  Kindly consult Appendix \ref{appb}.
\end{proof}

By using Lemma \ref{lemma-full}, constraints \eqref{eq:full-SINR-signal} can be expressed as
\begin{equation} \label{eq:lemma-full-1}
  \mathbf{x}_{k}^{\mathrm{H}}\widetilde{\mathbf{A}}_k\mathbf{x}_k+2\mathfrak{R} \{ \widetilde{\mathbf{e}}_{k}^{\mathrm{H}}\mathbf{x}_k \} +\widetilde{e}_k\ge \iota _k( 2^{R_{\mathrm{th}}}-1) ,\mathcal{I} _{k}^{f},\forall k\in \mathcal{K}.
\end{equation}

For the application of Lemma \ref{S-procedure}, expressing $\mathcal{I} _{k}^{f}$ concerning the quadratic expressions presented below, is beneficial:
\begin{equation*}
  \mathcal{I} _{k}^{f}\triangleq \begin{cases}
    \mathbf{x}_{k}^{\mathrm{H}}\left[ \begin{matrix}
    \mathbf{I}_M&		\mathbf{0}\\
    \mathbf{0}&		\mathbf{0}\\
  \end{matrix} \right] \mathbf{x}_k-\xi _{\mathrm{D},k}^{2}\le 0,\\
    \mathbf{x}_{k}^{\mathrm{H}}\left[ \begin{matrix}
    \mathbf{0}&		\mathbf{0}\\
    \mathbf{0}&		\mathbf{I}_{NM}\\
  \end{matrix} \right] \mathbf{x}_k-\xi _{\mathrm{H},k}^{2}\le 0.\\
  \end{cases}
\end{equation*}

Thus, through the use of slack variables $\boldsymbol{\chi }_{\mathrm{D}}=[ \chi _{\mathrm{D},1},\cdots ,\chi _{\mathrm{D},K} ] ^{\mathrm{T}}\ge 0$ and $\boldsymbol{\chi }_{\mathrm{H}}=[ \chi _{\mathrm{H},1},\cdots,\chi _{\mathrm{H},K} ] ^{\mathrm{T}}\ge 0$ and using Lemma \ref{S-procedure}, the constraints in \eqref{eq:full-SINR-signal} are expressed as the following equivalent LMIs:
\begin{equation} \label{eq:full-SINR-signal-LMI}
  \left[ \begin{matrix}
    \widetilde{\mathbf{A}}_k+\left[ \begin{matrix}
    \chi _{\mathrm{D},k}\mathbf{I}_M&		0\\
    0&		\chi _{\mathrm{H},k}\mathbf{I}_{NM}\\
     \end{matrix} \right]&		\widetilde{\mathbf{e}}_k\\
    \widetilde{\mathbf{e}}_{k}^{\mathrm{H}}&		C_{k}^{f}\\
  \end{matrix} \right] \succeq \mathbf{0},\forall k\in \mathcal{K},
\end{equation}
where $C_{k}^{f}=\widetilde{e}_k-\iota _k(2^{R_{\mathrm{th}}}-1)-\chi _{\mathrm{D},k}\xi _{\mathrm{D},k}^{2}-\chi _{\mathrm{H},k}\xi _{\mathrm{H},k}^{2}$.

In the following, we substitute $\mathbf{h}_{\mathrm{D},k}=\widehat{\mathbf{h}}_{\mathrm{D},k}+\bigtriangleup \mathbf{h}_{\mathrm{D},k}$ and $\mathbf{H}_k=\widehat{\mathbf{H}}_k+\bigtriangleup \mathbf{H}_k$ into the equivalent  matrix inequality representing the INs power in \eqref{eq:IN-LMI-1-1}, resulting in:
\begin{align}
  &\mathbf{0}  \preceq \left[ \begin{matrix}
    \iota _k-\sigma _{k}^{2}&	\widetilde{\mathbf{r}}_{k}^{\mathrm{H}}\\
    \widetilde{\mathbf{r}}_k&	\mathbf{I}\\
  \end{matrix} \right] \nonumber
  \\
  &+\left[ \begin{matrix}
    \mathbf{0}&		( \bigtriangleup \mathbf{h}_{\mathrm{D},k}^{\mathrm{H}}+\boldsymbol{\theta }^{\mathrm{H}}\bigtriangleup \mathbf{H}_k ) \mathbf{W}_{-k}\\
    \mathbf{W}_{-k}^{\mathrm{H}}( \bigtriangleup \mathbf{h}_{\mathrm{D},k}^{\mathrm{H}}+\bigtriangleup \mathbf{H}_{k}^{\mathrm{H}}\boldsymbol{\theta } )&	\mathbf{0}\\
  \end{matrix} \right] \nonumber
  \\
  &\preceq \left[ \begin{array}{c}
    \mathbf{0}\\
    \mathbf{W}_{-k}^{\mathrm{H}}\\
  \end{array} \right] \left[ \begin{matrix}
    \bigtriangleup \mathbf{h}_{\mathrm{D},k}&		\mathbf{0}\\
  \end{matrix} \right] +\left[ \begin{array}{c}
    \bigtriangleup \mathbf{h}_{\mathrm{D},k}^{\mathrm{H}}\\
    \mathbf{0}\\
  \end{array} \right] \left[ \begin{matrix}
    \mathbf{0}&		\mathbf{W}_{-k}\\
  \end{matrix} \right] \nonumber
  \\
  &+\left[ \begin{array}{c}
    \mathbf{0}\\
    \mathbf{W}_{-k}^{\mathrm{H}}\\
  \end{array} \right] \bigtriangleup \mathbf{H}_{k}^{\mathrm{H}}\left[ \begin{matrix}
    \boldsymbol{\theta }&		\mathbf{0}\\
  \end{matrix} \right] +\left[ \begin{array}{c}
    \boldsymbol{\theta }^{\mathrm{H}}\\
    \mathbf{0}\\
  \end{array} \right] \bigtriangleup \mathbf{H}_k\left[ \begin{matrix}
    \mathbf{0}&		\mathbf{W}_{-k}\\
  \end{matrix} \right] \nonumber
  \\
  &+\left[ \begin{matrix}
    \iota _k-\sigma _{k}^{2}&		\widetilde{\mathbf{r}}_{k}^{\mathrm{H}}\\
    \widetilde{\mathbf{r}}_k&		\mathbf{I}\\
  \end{matrix} \right],
\end{align}
where $\widetilde{\mathbf{r}}_k=( ( \mathbf{h}_{\mathrm{D},k}^{\mathrm{H}}+\boldsymbol{\theta }^{\mathrm{H}}\widehat{\mathbf{H}}_k ) \mathbf{W}_{-k} ) ^{\mathrm{H}}$.

After introducing slack variables $\boldsymbol{\eta }_{\mathrm{D}}=\left[ \eta _{\mathrm{D},1},\cdots ,\eta _{\mathrm{D},K} \right] ^{\mathrm{T}}\ge 0$ and $\boldsymbol{\eta }_{\mathrm{H}}=\left[ \eta _{\mathrm{H},1},\cdots ,\eta _{\mathrm{H},K} \right] ^{\mathrm{T}}\ge 0$. Using Lemma \ref{sign-definiteness}, the constraints in \eqref{eq:full-SINR-interference} are reformulated into the following equivalent LMIs:
\begin{equation} \label{eq:full-SINR-interference-LMI}
   \left[ \begin{matrix}
    \rho _k&		\widetilde{\mathbf{r}}_{k}^{\mathrm{H}}&		\mathbf{0}_{1\times M}&		\mathbf{0}_{1\times M}\\
    \widetilde{\mathbf{r}}_k&		\mathbf{I}_{(K-1)}&		\xi _{\mathrm{H},k}\mathbf{W}_{-k}^{\mathrm{H}}&		\xi _{\mathrm{D},k}\mathbf{W}_{-k}^{\mathrm{H}}\\
    \mathbf{0}_{M\times 1}&		\xi _{\mathrm{H},k}\mathbf{W}_{-k}&		\eta _{\mathrm{H},k}\mathbf{I}_M&		0_{M\times M}\\
    \mathbf{0}_{M\times 1}&		\xi _{\mathrm{D},k}\mathbf{W}_{-k}&		0_{M\times M}&		\eta _{\mathrm{D},k}\mathbf{I}_M\\
  \end{matrix} \right] \succeq \mathbf{0},\forall k\in \mathcal{K},
\end{equation}
where $\rho _k=\iota _k-\sigma _{k}^{2}-\eta _{\mathrm{H},k}N-\eta _{\mathrm{D},k}$.

Using \eqref{eq:full-SINR-signal-LMI} and \eqref{eq:full-SINR-interference-LMI}, the worst-case robust beamforming design problem under FCU is expressed as
\begin{subequations}
  \begin{align}   \label{eq:Pro6}
    \mathcal{P}_6:
    \max\limits _{{\scriptstyle {\mathbf{W},\boldsymbol{\theta },\boldsymbol{\iota },\boldsymbol{\chi }_{\mathrm{D}}\hfill\atop {\scriptstyle \boldsymbol{\chi }_{\mathrm{H}},\boldsymbol{\eta }_{\mathrm{D}},\boldsymbol{\eta }_{\mathrm{H}}  }}}} & \left\| \mathbf{W} \right\| _{F}^{2}     \\
    \mathrm{s}.\mathrm{t}. \thinspace\thinspace\thinspace\thinspace \ \ &  \eqref{eq:full-SINR-signal-LMI}, \eqref{eq:full-SINR-interference-LMI}, \eqref{eq:Pro1c},  \label{eq:Pro6b} \\
    & \boldsymbol{\chi }_{\mathrm{D}}\ge 0,\boldsymbol{\chi }_{\mathrm{H}}\ge 0,\boldsymbol{\eta }_{\mathrm{D}}\ge 0,\boldsymbol{\eta }_{\mathrm{H}}\ge 0.  \label{eq:Pro6c}
  \end{align}
\end{subequations}

$\mathcal{P}_6$ is also NC and involves coupled variables. Its solution follows a similar approach to that of  $\mathcal{P}_2$, and therefore, we omit the detailed derivation for simplicity.

\section{Outage Constrained Robust Beamforming Design For The SCSIE Model}  \label{ocrbd}
In traditional approaches, we often assume the CSI channel estimation error is follow a Gaussian distribution \cite{c14}, but it is actually unbounded in reality. Therefore, the BCSIE model can not accurately represent the real channel error. Therefore, this section introduces the SCSIE model.  In particular, let $\varepsilon _1,\dots,\varepsilon _K\in \left( 0,1 \right]$ be the maximum outage probabilities, and the transmit power minimization problem is given by
\begin{subequations}
  \begin{align}   \label{eq:Pro7}
    \mathcal{P}_7:
    \underset{\mathbf{W},\boldsymbol{\theta }\thinspace\thinspace}{\min} &  \left\| \mathbf{W} \right\| _{F}^{2}     \\
    \mathrm{s}.\mathrm{t}. & \thinspace\thinspace \mathrm{Pr}\left\{ R_k\left( \mathbf{W},\boldsymbol{\theta } \right) \ge R_{\mathrm{th}} \right\} \ge 1-\varepsilon _k,\forall k\in \mathcal{K},  \label{eq:Pro7b} \\
    & \left| \theta _n \right|^2=1,\forall n\in \mathcal{N}.  \label{eq:Pro7c}
  \end{align}
\end{subequations}
The outage constraints in \eqref{eq:Pro7b} are designed to control each user's OP, ensuring that the user can successfully receive its data transmission at a rate of $R_{\mathrm{th}}$ with a success probability no less than $1- \varepsilon_k$.

Solving $\mathcal{P}_7$ is difficult since the OP constraints in \eqref{eq:Pro7b} exhibit non-convexity, which lack simple closed-form solutions, as discussed in \cite{c15}. To overcome this problem, we employ the BTI as an approximation for the constraints, as described in the subsequent lemma.
\begin{lemma} \label{BTI} (Bernstein-Type Inequality (BTI): Lemma \ref{S-procedure} in \cite{c15})
Assume that $f\left( \mathbf{x} \right) =\mathbf{x}^{\mathrm{H}}\mathbf{Qx}+2\mathfrak{R} \left\{ \mathbf{q}^{\mathrm{H}}\mathbf{x} \right\} +q$, where $\mathbf{Q}\in \mathbb{H} ^{n\times n}$, $\mathbf{q}\in \mathbb{C} ^{n\times 1}$, $q\in \mathbb{R} $ and $\mathbf{x}\in \mathbb{C} ^{n\times 1}\sim \mathcal{C} \mathcal{N} (0,\mathbf{I})$. The following approximation holds for any $\varepsilon \in \left[ 0,1 \right] $:
\begin{subequations}\label{eq:outge-BIT}
  \begin{align}
   & \mathrm{Pr}\{\mathbf{x}^{\mathrm{H}}\mathbf{Qx}+2\mathfrak{R} \{\mathbf{q}^{\mathrm{H}}\mathbf{x}\}+q\ge 0\}\ge 1-\varepsilon  \label{eq:outge-BIT-a}\\
  \Rightarrow & \thinspace \mathrm{Pr}\{\mathbf{Q}\}-\sqrt{2\ln\mathrm{(}1/\varepsilon )}x+\ln\mathrm{(}\varepsilon )\lambda _{\max}^{+}(-\mathbf{Q})+q\ge 0 \label{eq:outge-BIT-b}  \\
  \Rightarrow & \left\{ \begin{array}{c} \label{eq:outge-BIT-c}\
    \mathrm{Tr}\{ \mathbf{Q} \} -\sqrt{2\ln\mathrm{(}1/\varepsilon )}x+\ln\mathrm{(}\varepsilon )y+q\ge 0,\\
    \sqrt{||\mathbf{Q}||_{F}^{2}+2||\mathbf{q}||^2}\le x,\\
    y\mathbf{I}+\mathbf{Q}\succeq 0,y\ge 0,
  \end{array}\right.
  \end{align}
  where $\lambda _{\max}^{+}(-\mathbf{Q})=\max ( \lambda _{\max}^{+}(-\mathbf{Q}),0 )$. $x$ and $y$ denote the introduced slack variables.
\end{subequations}
\end{lemma}


In the next subsections, we begin by presenting a simplified method for robust beamforming in the PCU scenario with SCSIE. Next, we adapt it for the FCU case.

\subsection{Scenario 1: PCU}
By exploiting the BTI, the  OP constraint of the $k$th user in \eqref{eq:Pro7b} can be transformed as
\begin{align} \label{eq:partial-outage}
 & \mathrm{Pr}\left\{ \log _2\left( 1+\frac{\left| \left( \mathbf{h}_{\mathrm{D},k}^{\mathrm{H}}+\boldsymbol{\theta }^{\mathrm{H}}\mathbf{H}_k \right) \mathbf{w}_k \right|^2}{\left\| \left( \mathbf{h}_{\mathrm{D},k}^{\mathrm{H}}+\boldsymbol{\theta }^{\mathrm{H}}\mathbf{H}_k \right) \mathbf{W}_{-k} \right\| _{2}^{2}+\sigma _{k}^{2}} \right) \ge R_{\mathrm{th}} \right\} \nonumber
\\
& =\mathrm{Pr}\{ ( \mathbf{h}_{\mathrm{D},k}^{\mathrm{H}}+\boldsymbol{\theta }^{\mathrm{H}}\mathbf{H}_k ) \mathbf{\Xi }_k( \mathbf{h}_{\mathrm{D},k}+\mathbf{H}_{k}^{\mathrm{H}}\boldsymbol{\theta } ) -\sigma _{k}^{2}\ge 0 \},
\end{align}
where $\mathbf{\Xi }_k=\mathbf{w}_k\mathbf{w}_{k}^{\mathrm{H}}/( 2^{R_{\mathrm{th}}}-1 ) -\mathbf{W}_{-k}\mathbf{W}_{-k}^{\mathrm{H}}$.

For simplicity, we assume that $\mathbf{\Sigma }_{\mathrm{H},k}=\delta _{\mathrm{H},k}^{2}\mathbf{I}$, and rewrite the CBRUB error in \eqref{eq:statistical} as $\mathrm{vec}( \bigtriangleup \mathbf{H}_k ) =\delta _{\mathrm{H},k}\mathbf{i}_{\mathrm{H},k}$, where $\mathbf{i}_{\mathrm{H},k}\in \mathcal{C} \mathcal{N} (0,\mathbf{I})$. By defining $\mathbf{\Theta }=\boldsymbol{\theta \theta }^{\mathrm{H}}$ and using $\mathbf{H}_k=\widehat{\mathbf{H}}_k+\bigtriangleup \mathbf{H}_k$, the OP given in \eqref{eq:partial-outage} can be equivalently transformed to \eqref{eq:partial-outage-1}. 
\begin{figure*}
  \begin{align}  \label{eq:partial-outage-1}
    &\mathrm{Pr}\left\{ \left( \mathbf{h}_{\mathrm{D},k}^{\mathrm{H}}+\boldsymbol{\theta }^H\left( \widehat{\mathbf{H}}_k+\bigtriangleup \mathbf{H}_k \right) \right) \mathbf{\Xi }_k\left( \mathbf{h}_{\mathrm{D},k}+\left( \widehat{\mathbf{H}}_k+\bigtriangleup \mathbf{H}_k \right) ^{\mathrm{H}}\boldsymbol{\theta } \right) -\sigma _{k}^{2}\ge 0 \right\} \nonumber   \\
    =& \thinspace \mathrm{Pr}\left\{  {\boldsymbol{\theta }^H\bigtriangleup \mathbf{H}_k\mathbf{\Xi }_k\bigtriangleup \mathbf{H}_{k}^{\mathrm{H}}\boldsymbol{\theta }} +2\mathrm{Re}\left\{  {\boldsymbol{\theta }^H\bigtriangleup \mathbf{H}_k\mathbf{\Xi }_k\left( \mathbf{h}_{\mathrm{D},k}+\widehat{\mathbf{H}}_{k}^{\mathrm{H}}\boldsymbol{\theta } \right) } \right\}  + {{\left( \mathbf{h}_{\mathrm{D},k}^{\mathrm{H}}+\boldsymbol{\theta }^{\mathrm{H}}\widehat{\mathbf{H}}_k \right) \mathbf{\Xi }_k\left( \mathbf{h}_{\mathrm{D},k}+\widehat{\mathbf{H}}_{k}^{\mathrm{H}}\boldsymbol{\theta } \right) -\sigma _{k}^{2}}}\ge 0 \right\}. 
  \end{align}
  \hrule 
\end{figure*}

The first term of the right-hand side of \eqref{eq:partial-outage-1} can be reformulated as 
\begin{align} \label{eq:f2}
  &{\boldsymbol{\theta }^{\mathrm{H}}\bigtriangleup \mathbf{H}_k\mathbf{\Xi }_k\bigtriangleup \mathbf{H}_{k}^{\mathrm{H}}\boldsymbol{\theta }} \nonumber \\
  =& \thinspace \mathrm{Tr}\left( \bigtriangleup \mathbf{H}_k\mathbf{\Xi }_k\bigtriangleup \mathbf{H}_{k}^{\mathrm{H}}\boldsymbol{\theta \theta }^{\mathrm{H}} \right) \nonumber \\
   =& \thinspace \mathrm{Tr}\left( \bigtriangleup \mathbf{H}_{k}^{\mathrm{H}}\mathbf{\Theta }\bigtriangleup \mathbf{H}_k\mathbf{\Xi }_k \right) \nonumber   \\
    \overset{\left( \mathrm{a} \right)}{=} &  \thinspace \mathrm{vec}^{\mathrm{H}}\left( \bigtriangleup \mathbf{H}_k \right) \left( \mathbf{\Xi }_{k}^{\mathrm{T}}\otimes \mathbf{\Theta } \right) \mathrm{vec}\left( \bigtriangleup \mathbf{H}_k \right),
\end{align}
where $\left( \mathrm{a} \right)$ is obtained by invoking the identity $\mathrm{Tr(}\mathbf{A}^{\mathrm{H}}\mathbf{BCD})=\mathrm{vec}^{\mathrm{H}}(\mathbf{A})(\mathbf{D}^{\mathrm{T}}\otimes \mathbf{B})\mathrm{vec(}\mathbf{C})$ \cite{c16,c18}. The second term of the right-hand side of \eqref{eq:partial-outage-1}  can be reformulated as
\begin{align} \label{eq:f3}
  &{\boldsymbol{\theta }^{\mathrm{H}}\bigtriangleup \mathbf{H}_k\mathbf{\Xi }_k\left( \mathbf{h}_{\mathrm{D},k}+\widehat{\mathbf{H}}_{k}^{\mathrm{H}}\boldsymbol{\theta } \right) } \nonumber  \\
  =& \thinspace\mathrm{Tr}\left( \bigtriangleup \mathbf{H}_k\mathbf{\Xi }_k\left( \mathbf{h}_{\mathrm{D},k}\boldsymbol{\theta }^{\mathrm{H}}+\widehat{\mathbf{H}}_{k}^{\mathrm{H}}\boldsymbol{\theta \theta }^{\mathrm{H}} \right) \right)  \nonumber \\
  = &\thinspace\mathrm{Tr}\left( \bigtriangleup \mathbf{H}_k\mathbf{\Xi }_k\left( \mathbf{h}_{\mathrm{D},k}\boldsymbol{\theta }^{\mathrm{H}}+\widehat{\mathbf{H}}_{k}^{\mathrm{H}}\mathbf{\Theta } \right) \right) \nonumber \\
   \overset{\left( b \right)}{=} & \thinspace \mathrm{vec}^{\mathrm{H}}\left( \boldsymbol{\theta }\mathbf{h}_{\mathrm{D},k}^{\mathrm{H}}+\mathbf{\Theta }\widehat{\mathbf{H}}_k \right) \left( \mathbf{\Xi }_{k}^{\mathrm{T}}\otimes \mathbf{I}_N \right) \mathrm{vec}\left( \bigtriangleup \mathbf{H}_k \right) \nonumber \\
 = &\thinspace\mathrm{vec}^{\mathrm{H}}\left( \left( \boldsymbol{\theta }\mathbf{h}_{\mathrm{D},k}^{\mathrm{H}}+\mathbf{\Theta }\widehat{\mathbf{H}}_k \right) \mathbf{\Xi }_k \right) \mathrm{vec}\left( \bigtriangleup \mathbf{H}_k \right),
\end{align}
where $\left( \mathrm{b} \right)$ is obtained by invoking the identity $\mathrm{Tr(}\mathbf{ABC}^{\mathrm{H}})=\mathrm{vec}^{\mathrm{H}}(\mathbf{C})(\mathbf{B}^{\mathrm{T}}\otimes \mathbf{I})\mathrm{vec(}\mathbf{A})$ \cite{c16,c18}. According to the above derivation, \eqref{eq:partial-outage-1} can be rewritten as in \eqref{eq:partial-outage-2}.
\begin{figure*}
  \begin{align}  \label{eq:partial-outage-2}
    & \mathrm{Pr}\left\{ \left( \mathbf{h}_{\mathrm{D},k}^{\mathrm{H}}+\boldsymbol{\theta }^{\mathrm{H}}\left( \widehat{\mathbf{H}}_k+\bigtriangleup \mathbf{H}_k \right) \right) \mathbf{\Xi }_k\left( \mathbf{h}_{\mathrm{D},k}+\left( \widehat{\mathbf{H}}_k+\bigtriangleup \mathbf{H}_k \right) ^{\mathrm{H}}\boldsymbol{\theta } \right) -\sigma _{k}^{2}\ge 0 \right\} \nonumber \\
     =& \thinspace\mathrm{Pr}\left\{ \mathrm{vec}^{\mathrm{H}}\left( \bigtriangleup \mathbf{H}_k \right) \left( \mathbf{\Xi }_{k}^{\mathrm{T}}\otimes \mathbf{\Theta } \right) \mathrm{vec}\left( \bigtriangleup \mathbf{H}_k \right) +2\mathrm{Re}\left\{ \mathrm{vec}^{\mathrm{H}}\left( \left( \boldsymbol{\theta }\mathbf{h}_{\mathrm{D},k}^{\mathrm{H}}+\mathbf{\Theta }\widehat{\mathbf{H}}_k \right) \mathbf{\Xi }_k \right) \mathrm{vec}\left( \bigtriangleup \mathbf{H}_k \right) \right\} \right. \nonumber \\
    & \thinspace\thinspace\thinspace\thinspace\thinspace\thinspace\thinspace+\left. \left( \mathbf{h}_{\mathrm{D},k}^{\mathrm{H}}+\boldsymbol{\theta }^{\mathrm{H}}\widehat{\mathbf{H}}_k \right) \mathbf{\Xi }_k\left( \mathbf{h}_{\mathrm{D},k}+\widehat{\mathbf{H}}_{k}^{\mathrm{H}}\boldsymbol{\theta } \right) -\sigma _{k}^{2}\ge 0 \right\}  \nonumber  \\
    =& \thinspace \mathrm{Pr}\left\{ \delta _{\mathrm{H},k}^{2}\mathbf{i}_{\mathrm{H},k}^{\mathrm{H}}\left( \mathbf{\Xi }_{k}^{\mathrm{T}}\otimes \mathbf{\Theta } \right) \mathbf{i}_{\mathrm{H},k}+2\mathrm{Re}\left\{ \delta _{\mathrm{H},k}\mathrm{vec}^{\mathrm{H}}\left( \left( \boldsymbol{\theta }\mathbf{h}_{\mathrm{D},k}^{\mathrm{H}}+\mathbf{\Theta }\widehat{\mathbf{H}}_k \right) \mathbf{\Xi }_k \right) \mathbf{i}_{\mathrm{H},k} \right\} \right. \nonumber \\
    & \thinspace\thinspace\thinspace\thinspace\thinspace\thinspace\thinspace \left. +\left( \mathbf{h}_{\mathrm{D},k}^{\mathrm{H}}+\boldsymbol{\theta }^{\mathrm{H}}\widehat{\mathbf{H}}_k \right) \mathbf{\Xi }_k\left( \mathbf{h}_{\mathrm{D},k}+\widehat{\mathbf{H}}_{k}^{\mathrm{H}}\boldsymbol{\theta } \right) -\sigma _{k}^{2}\ge 0 \right\},
  \end{align}
  \hrule 
\end{figure*}

Thus, the OP constraints \eqref{eq:Pro7b} can be reformulated as
\begin{equation} \label{eq:partial-outage-3}
  \mathrm{Pr}\{ \mathbf{i}_{\mathrm{H},k}^{\mathrm{H}}\mathbf{Q}_k\mathbf{i}_{\mathrm{H},k}+2\mathfrak{R} \{ \mathbf{q}_{k}^{\mathrm{H}}\mathbf{i}_{\mathrm{H},k} \}+q_k\ge 0 \} \ge 1-\varepsilon _k,\forall k\in \mathcal{K},
\end{equation}
where
\begin{subequations}
\begin{align}
  &\mathbf{Q}_k=\delta _{\mathrm{H},k}^{2}( \mathbf{\Xi }_{k}^{\mathrm{T}}\otimes \mathbf{\Theta } , \label{eq:Q}  \\
  &\mathbf{q}_k=\delta _{\mathrm{H},k}\mathrm{vec}( ( \boldsymbol{\theta }\mathbf{h}_{\mathrm{D},k}^{\mathrm{H}}+\mathbf{\Theta }\widehat{\mathbf{H}}_k ) \mathbf{\Xi }_{k}^{\mathrm{H}} ), \label{eq:q1} \\
  &q_k=( \mathbf{h}_{\mathrm{D},k}^{\mathrm{H}}+\boldsymbol{\theta }^{\mathrm{H}}\widehat{\mathbf{H}}_k ) \mathbf{\Xi }_k( \mathbf{h}_{\mathrm{D},k}+\widehat{\mathbf{H}}_{k}^{\mathrm{H}}\boldsymbol{\theta }) -\sigma _{k}^{2}. \label{eq:q2}
\end{align}
\end{subequations}

Using Lemma \ref{BTI} and defining auxiliary variables $\mathbf{x}=\left[ x_1,\cdots ,x_K \right] ^{\mathrm{T}}$ and $\mathbf{y}=\left[ y_1,\cdots ,y_K \right] ^{\mathrm{T}}$, the OP of user $k$ in \eqref{eq:partial-outage-3} can be given in the deterministic form as
\begin{subequations} \label{eq:BTI}
  \begin{align}
    &\mathrm{Tr}\left\{ \mathbf{Q}_k \right\} -\sqrt{2\ln\mathrm{(}1/\varepsilon _k)}x_k+\ln\mathrm{(}\varepsilon _k)y_k+q_k\ge 0, \label{eq:BTI1}  \\
    &\sqrt{||\mathbf{Q}_k||_{F}^{2}+2||\mathbf{q}_k|||^2}\le x_k,  \\
    &y_k\mathbf{I}+\mathbf{Q}_k\succeq \mathbf{0},y_k\ge 0. 
  \end{align}
\end{subequations}

By utilizing  $\mathrm{Tr}(\mathbf{A}\otimes\mathbf{B})=\mathrm{Tr}(\mathbf{A})\mathrm{Tr}(\mathbf{B})$ and $\mathbf{AB}\otimes\mathbf{CD}=(\mathbf{A}\otimes\mathbf{C})(\mathbf{B}\otimes\mathbf{D})\mathrm{~and~}\mathbf{a}\otimes\mathbf{b}=\mathrm{vec}(\mathbf{ba}^{\mathrm{T}})$ \cite{c18}, \eqref{eq:BTI} can be simplified further through the following mathematical transformations:
\begin{subequations} \label{eq:BTI2}
  \begin{align}
    \mathrm{Tr}\{ \mathbf{Q}_k \} &=\delta _{\mathrm{H},k}^{2}\mathrm{Tr}\{ \mathbf{\Xi }_{k}^{\mathrm{T}}\otimes \mathbf{\Theta } \} =\delta _{\mathrm{H},k}^{2}\mathrm{Tr}\{ \mathbf{\Xi }_k \} \mathrm{Tr}\{ \mathbf{\Theta } \}  \nonumber \\
    &=\delta _{\mathrm{H},k}^{2}N\mathrm{Tr}\{ \mathbf{\Xi }_k \}, \label{eq:BTI2-1}  \\
    \| \mathbf{Q}_k \| _{F}^{2}&=\delta _{\mathrm{H},k}^{4}\| \mathbf{\Xi }_{k}^{\mathrm{T}}\otimes \mathbf{\Theta } \| _{F}^{2}=\delta _{\mathrm{H},k}^{4}\| \mathbf{\Xi }_k \| _{F}^{2}\| \mathbf{\Theta } \| _{F}^{2}
    \nonumber   \\
    & =\delta _{\mathrm{H},k}^{4}N^2\| \mathbf{\Xi }_k \| _{F}^{2}, \label{eq:BTI2-2} \\
    \| \mathbf{q}_k \| ^2&=\delta _{\mathrm{H},k}^{2}\| \mathrm{vec}( ( \boldsymbol{\theta }\mathbf{h}_{\mathrm{D},k}^{\mathrm{H}}+\mathbf{\Theta }\widehat{\mathbf{H}}_k ) \mathbf{\Xi }_{k}^{\mathrm{H}} ) \| ^2
    \nonumber  \\
    &=\delta _{\mathrm{H},k}^{2}N\| ( \mathbf{h}_{\mathrm{D},k}^{\mathrm{H}}+\boldsymbol{\theta }^{\mathrm{H}}\widehat{\mathbf{H}}_k ) \mathbf{\Xi }_k \| _{2}^{2}, \label{eq:BTI2-3}  \\
    \lambda ( \mathbf{Q}_k)&  =\lambda ( \delta _{\mathrm{H},k}^{2}( \mathbf{\Xi }_{k}^{\mathrm{T}}\otimes \mathbf{\Theta } ) ) =\delta _{\mathrm{H},k}^{2}\lambda ( \mathbf{\Xi }_{k}^{\mathrm{T}}\otimes \mathbf{\Theta } )
     \nonumber \\
    &=\delta _{\mathrm{H},k}^{2}\lambda ( \mathbf{\Xi }_k ) \lambda ( \mathbf{\Theta } ) =\delta _{\mathrm{H},k}^{2}N\lambda ( \mathbf{\Xi }_k ), \label{eq:BTI2-4}
  \end{align}
\end{subequations}
where $\lambda ( \mathbf{X} ) $ denotes the eigenvalues of $\mathbf{X}$.

Hence, leveraging Lemma \ref{BTI} and equation \eqref{eq:BTI2}, we can represent the approximation $\mathcal{P}_7$ as:
\begin{subequations}
  \begin{align}   \label{eq:Pro8}
    &  \mathcal{P}_8:
    \underset{\mathbf{W},\boldsymbol{\theta },\mathbf{x},\mathbf{y}}{\min}
    \left\| \mathbf{W} \right\| _{F}^{2}     \\
    & \mathrm{s}.\mathrm{t}. \thinspace\thinspace\thinspace\thinspace   \delta _{\mathrm{H},k}^{2}N\mathrm{Tr}\left\{ \mathbf{\Xi }_k \right\} -\sqrt{2\ln\mathrm{(}1/\varepsilon _k)}x_k  \nonumber    \\
    & \ \ \ \ \ -\ln\mathrm{(}1/\varepsilon _k)y_k   +q_k\ge 0,\forall k\in \mathcal{K}, \label{eq:Pro8b} \\
    & \ \ \ \ \ \left\| \begin{array}{c}
      \delta _{\mathrm{H},k}^{2}N\mathrm{vec}\left( \mathbf{\Xi }_k \right)\\
      \sqrt{2N}\delta _{\mathrm{H},k}\mathbf{\Xi }_k( \mathbf{h}_{\mathrm{D},k}+\widehat{\mathbf{H}}_{k}^{\mathrm{H}}\boldsymbol{\theta } )\\
    \end{array} \right\| \le x_{k,}\forall k\in \mathcal{K},  \label{eq:Pro8c} \\
    & \ \ \ \ \  y_k\mathbf{I}+\mathbf{Q}_k\succeq \mathbf{0},y_k\ge 0,\forall k\in \mathcal{K},  \label{eq:Pro8d} \\
    &  \ \ \ \ \ \left| \theta _n \right|^2=1,\forall n\in \mathcal{N}. \label{eq:Pro8e}
  \end{align}
\end{subequations}

This problem remains difficult to address since the constraints \eqref{eq:Pro8c} and  \eqref{eq:Pro8e} exhibit non-convexity, which involve coupled variables $\mathbf{W}$ and $\boldsymbol{\theta }$. To overcome this, we employ the AO approach to iteratively update $\mathbf{W}$ and $\boldsymbol{\theta }$.  In particular, when $\boldsymbol{\theta}$ is fixed, we relax the NC problem with respect to $\mathbf{W}$ the SDR technique \cite{c17}, solving it with CVX. Afterward, $\mathbf{W}$ is given, and the resulting problem that is NC for $\boldsymbol{\theta}$ is also addressed using the SDR technique.

For fixed $\boldsymbol{\theta}$, defining $\mathbf{\Xi }_k=\mathbf{\Psi }_k/( 2^{R_{\mathrm{th}}}-1 ) -\sum\nolimits_{i=1,i\ne k}^K{\mathbf{\Psi }_i}$, where $\mathbf{\Psi }_k=\mathbf{w}_k\mathbf{w}_{k}^{\mathrm{H}}$, the problem $\mathcal{P}_8$ associated with $\mathbf{W}$ is re-expressed as:
\begin{subequations}
  \begin{align}   \label{eq:Pro9}
    &  \mathcal{P}_9:
    \underset{\mathbf{\Psi },\mathbf{x},\mathbf{y}\,\,}{\min} \sum_{k=1}^K{\mathrm{Tr}\left\{ \mathbf{\Psi }_k \right\}}     \\
    & \mathrm{s}.\mathrm{t}. \thinspace\thinspace\thinspace\thinspace   \delta _{\mathrm{H},k}^{2}N\mathrm{Tr}\left\{ \mathbf{\Xi }_k \right\} -\sqrt{2\ln\mathrm{(}1/\varepsilon _k)}x_k
     \nonumber    \\
    & \ \ \ \ -\ln\mathrm{(}1/\varepsilon _k)y_k+q_k\ge 0,\forall k\in \mathcal{K}, \label{eq:Pro9b} \\
    & \ \ \ \  \left\| \begin{array}{c}
      \delta _{\mathrm{H},k}^{2}N\mathrm{vec}\left( \mathbf{\Xi }_k \right)\\
      \sqrt{2N}\delta _{\mathrm{H},k}\mathbf{\Xi }_k( \mathbf{h}_{\mathrm{D},k}+\widehat{\mathbf{H}}_{k}^{\mathrm{H}}\boldsymbol{\theta } )\\
    \end{array} \right\| \le x_{k,}\forall k\in \mathcal{K},  \label{eq:Pro9c} \\
    & \ \ \ \  y_k\mathbf{I}+\mathbf{Q}_k\succeq \mathbf{0},y_k\ge 0,\forall k\in \mathcal{K},  \label{eq:Pro9d} \\
    &  \ \ \ \  \left| \theta _n \right|^2=1,\forall n\in \mathcal{N}, \label{eq:Pro9e} \\
    & \ \ \ \  \mathbf{\Psi }_k\succeq \mathbf{0},\forall k\in \mathcal{K},  \label{eq:Pro9f} \\
    & \ \ \ \  \mathrm{rank}( \mathbf{\Psi }_k ) =1,\forall k\in \mathcal{K}. \label{eq:Pro9g}
  \end{align}
\end{subequations}
where $\mathbf{\Psi }=\left[ \mathbf{\Psi }_1,\dots ,\mathbf{\Psi }_K \right] $.  $\mathcal{P}_9$ is addressed by applying the SDR technique, where constraints \eqref{eq:Pro9f} are removed from the formulation. The resulting  problem that is convex SDP can be addressed efficiently via the CVX tools. The subsequent theorem shows the tightness of the SDR.
\begin{theorem} \cite{r16} \label{theorem-partical} If it is feasible for the relaxed form of $\mathcal{P}_9$, a solution that satisfies feasibility always exists, denoted as $\mathbf{\Psi }^{\star}=\left[ \mathbf{\Psi }_{1}^{\star},\dots ,\mathbf{\Psi }_{K}^{\star} \right] $, where $\mathrm{rank}\left( \mathbf{\Psi }_{k}^{\star} \right) =1,\forall k\in \mathcal{K}$.
\end{theorem}

\begin{proof}
  Kindly consult Appendix \ref{appc}.
\end{proof}

\begin{remark} \cite{r16} \label{remark-partical}
  The simulation analysis demonstrate that the optimal $\mathbf{\Psi }_{k}^{\star}$ generally has a rank of one before the rank-1 solution is constructed as described in Appendix \ref{appc}. The optimal $\mathbf{w}_k$ can be obtained from $\mathbf{\Psi }_{k}^{\star}$ through eigenvalue decomposition.
\end{remark}

For fixed $\mathbf{W}$, the subproblem involving $\boldsymbol{\theta}$ turns into a feasibility-checking problem. To refine the optimized solution for $\boldsymbol{\theta}$, as indicated in \cite{c11, c12}, we adjust the OP in \eqref{eq:partial-outage} by adding slack variable  $\widetilde{\boldsymbol{\alpha}}=\left[ \widetilde{\alpha}_1,\dots ,\widetilde{\alpha}_K \right] ^{\mathrm{T}}\ge 0$, which can be reformulated as:
\begin{equation*}
  \mathrm{Pr}\{ ( \mathbf{h}_{\mathrm{D},k}^{\mathrm{H}}+\boldsymbol{\theta }^{\mathrm{H}}\mathbf{H}_k ) \mathbf{\Xi }_k( \mathbf{h}_{\mathrm{D},k}+\mathbf{H}_{k}^{\mathrm{H}}\boldsymbol{\theta } ) -\sigma _{k}^{2}-\widetilde{{\alpha}}_{k}\ge 0 \}.
\end{equation*}
Then, \eqref{eq:q2} is also rewritten as
\begin{align} \label{eq:q2-new}
  q_{k}^{\boldsymbol{\theta }}&=( \mathbf{h}_{\mathrm{D},k}^{\mathrm{H}}+\boldsymbol{\theta }^{\mathrm{H}}\widehat{\mathbf{H}}_k ) \mathbf{\Xi }_k( \mathbf{h}_{\mathrm{D},k}+\widehat{\mathbf{H}}_{k}^{\mathrm{H}}\boldsymbol{\theta } ) -\sigma _{k}^{2}-\widetilde{{\alpha}} _k \nonumber \\
  &=\mathrm{Tr}\{ \mathbf{C}_k\widehat{\mathbf{\Theta}} \} +\mathbf{h}_{\mathrm{D},k}^{\mathrm{H}}\mathbf{\Xi }_k\mathbf{h}_{\mathrm{D},k}-\sigma _{k}^{2}-\widetilde{{\alpha}} _k,
\end{align}
where
\begin{equation*}
  \mathbf{C}_k=\left[ \begin{matrix}
    \widehat{\mathbf{H}}_k\mathbf{\Xi }_k\widehat{\mathbf{H}}_{k}^{\mathrm{H}}&		\widehat{\mathbf{H}}_k\mathbf{\Xi }_k\mathbf{h}_k\\
    \mathbf{h}_{k}^{\mathrm{H}}\mathbf{\Xi }_k\widehat{\mathbf{H}}_{k}^{\mathrm{H}}&		\mathbf{0}\\
  \end{matrix} \right],\widehat{\mathbf{\Theta}}=\widehat{\boldsymbol{\theta}}\widehat{\boldsymbol{\theta}}^{\mathrm{H}},\widehat{\boldsymbol{\theta}}=\left[ \begin{array}{c}
    \boldsymbol{\theta }\\
    1\\
  \end{array} \right].
\end{equation*}

With variable $\widehat{\mathbf{\Theta}}$, \eqref{eq:Pro8c} can be modified as follows
\begin{align} \label{eq:Pro8c-new}
  &\delta _{\mathrm{H},k}^{4}N^2||\mathbf{\Xi }_k||_{F}^{2}+2\delta _{\mathrm{H},k}^{2}N( \mathrm{Tr}\{ \mathbf{R}_{\mathrm{th}}\widehat{\mathbf{\Theta}} \} +\mathbf{h}_{\mathrm{D},k}^{\mathrm{H}}\mathbf{\Xi }_{k}^{\mathrm{H}}\mathbf{\Xi }_k\mathbf{h}_{\mathrm{D},k} ) \nonumber \\
  &\le 2\mathfrak{R} \{x_{k}^{(n)}x_k\}-x_{k}^{(n),2},\forall k\in \mathcal{K},
\end{align}
where $x_{k}^{(n)}$ denotes the optimal solution of the $n$-th iteration.

Moreover, constraints \eqref{eq:Pro8d} do not depend on $\boldsymbol{\theta}$ and are derived from $\lambda _{\max}^{+}( -\mathbf{Q} ) $ as presented in Lemma \ref{BTI}. As a result, we can express $y_k=\max ( \lambda _{\max}( -\delta _{\mathrm{H},k}^{2}N\mathbf{\Xi }_k ) ,0 ) ,\forall k\in \mathcal{K}$.

As a result, the reflection beamforming subproblem associated with $\boldsymbol{\theta}$ of $\mathcal{P}_8$ can be formulated as:
\begin{subequations}
  \begin{align}   \label{eq:Pro10}
    &\mathcal{P}_{10}:
    \underset{\widehat{\mathbf{\Theta}},\widetilde{\boldsymbol{\alpha}},\mathbf{x}}{\max} \sum_{k=1}^K{\widetilde{\alpha }}_k     \\
    &  \mathrm{s}.\mathrm{t}. \thinspace\thinspace\thinspace\thinspace  \delta _{\mathrm{H},k}^{2}N\mathrm{Tr}\left\{ \mathbf{\Xi }_k \right\} -\sqrt{2\ln\mathrm{(}1/\varepsilon _k)}x_k
     \nonumber  \\
   & \ \ \ \ \ -\ln\mathrm{(}1/\varepsilon _k)y_k+q_{k}^{\boldsymbol{\theta }}\ge 0,\forall k\in \mathcal{K}, \label{eq:Pro10b} \\
    & \ \ \ \ \ \ \eqref{eq:Pro8c-new}, \widetilde{\boldsymbol{\alpha}} \ge 0, \label{eq:Pro10c} \\
    &  \ \ \ \ \ \ \widehat{\mathbf{\Theta}}\succeq \mathbf{0},\mathrm{rank}( \widehat{\mathbf{\Theta}} ) =1,[ \widehat{\mathbf{\Theta}} ] _{k,k}=1,\forall k\in \mathcal{K}. \label{eq:Pro10d}
  \end{align}
\end{subequations}

By applying the SDR method, the convex SDP problem corresponding to the relaxed version of $\mathcal{P}_{10}$ is addressed via the CVX tools. The optimal $\boldsymbol{\theta }$ is  obtained from the optimal $\widehat{\mathbf{\Theta}}^{\star}$ through Gaussian randomization approaches \cite{c17}.

\vspace{-3mm}
\subsection{Scenario 2: FCU}
In this subsection, we broaden the OP constrained robust beamforming design from the PCU scenario to the FCU scenario. By taking into account the full SCSIE in \eqref{eq:statistical}, \eqref{eq:partial-outage} is subsequently rewritten as
\begin{align} \label{eq:full-outage}
  &\mathrm{Pr}\left\{ ( \widehat{\mathbf{h}}_{\mathrm{D},k}^{\mathrm{H}}+\boldsymbol{\theta }^{\mathrm{H}}\widehat{\mathbf{H}}_k ) \mathbf{\Xi }_k( \widehat{\mathbf{h}}_{\mathrm{D},k}+\widehat{\mathbf{H}}_{k}^{\mathrm{H}}\boldsymbol{\theta } ) \right.
  \nonumber      \\
  &+2\mathfrak{R} \{ ( \widehat{\mathbf{h}}_{\mathrm{D},k}^{\mathrm{H}}+\boldsymbol{\theta }^{\mathrm{H}}\widehat{\mathbf{H}}_k ) \mathbf{\Xi }_k( \bigtriangleup \mathbf{h}_{\mathrm{D},k}+\bigtriangleup \mathbf{H}_{k}^{\mathrm{H}}\boldsymbol{\theta } ) \} -\sigma _{k}^{2} \nonumber \\
  &+\left. ( \bigtriangleup \mathbf{H}_{k}^{\mathrm{H}}+\boldsymbol{\theta }^{\mathrm{H}}\bigtriangleup \mathbf{H}_k ) \mathbf{\Xi }_k\bigl( \bigtriangleup \mathbf{h}_{\mathrm{D},k}+\bigtriangleup \mathbf{H}_{k}^{\mathrm{H}}\boldsymbol{\theta } \bigr) \ge 0 \right\}.
\end{align}
Assuming $\mathbf{\Sigma }_{\mathrm{D},k}=\delta _{\mathrm{D},k}^{2}\mathbf{I}$, the DCSIB can be written as $\bigtriangleup \mathbf{h}_k=\delta _{\mathrm{D},k}\mathbf{i}_{\mathrm{D},k}$, where $\mathbf{i}_{\mathrm{D},k}\in \mathcal{C} \mathcal{N} (0,\mathbf{I})$. The second component in \eqref{eq:full-outage}  can be then reformulated as
\begin{align*}
  &2\mathfrak{R} \left\{ ( \widehat{\mathbf{h}}_{\mathrm{D},k}^{\mathrm{H}}+\boldsymbol{\theta }^{\mathrm{H}}\widehat{\mathbf{H}}_k ) \mathbf{\Xi }_k\bigtriangleup \mathbf{h}_{\mathrm{D},k} \right. \\
  & \ \ \ \ \left. +  \mathrm{vec}^{\mathrm{T}}( \boldsymbol{\theta }( \widehat{\mathbf{h}}_{\mathrm{D},k}^{\mathrm{H}}+\boldsymbol{\theta }^{\mathrm{H}}\widehat{\mathbf{H}}_k ) \mathbf{\Xi }_k ) \mathrm{vec}( \bigtriangleup \mathbf{H}_{k}^{*} ) \right\}
  \\
  =&\thinspace 2 \mathfrak{R} \left\{ \delta _{\mathrm{D},k}( \widehat{\mathbf{h}}_{\mathrm{D},k}^{\mathrm{H}}+\boldsymbol{\theta }^{\mathrm{H}}\widehat{\mathbf{H}}_k) \mathbf{\Xi }_k\mathbf{i}_{\mathrm{D},k} \right. \\
  & \ \ \ \ \left.
  +\delta _{\mathrm{H},k}\mathrm{vec}^{\mathrm{T}}( \boldsymbol{\theta }( \widehat{\mathbf{h}}_{\mathrm{D},k}^{\mathrm{H}} +\boldsymbol{\theta }^{\mathrm{H}}\widehat{\mathbf{H}}_k ) \mathbf{\Xi }_k ) \mathbf{i}_{\mathrm{H},k}^{*} \right\} \\
  =&\thinspace2\mathfrak{R}\{ \widetilde{\mathbf{q}}_{k}^{\mathrm{H}}\widetilde{\mathbf{i}}_k \},
\end{align*}
where $\widetilde{\mathbf{i}}_k=\left[ \begin{matrix}
	\mathbf{i}_{\mathrm{D},k}^{\mathrm{T}}&		\mathbf{i}_{\mathrm{H},k}^{\mathrm{T}}\\
\end{matrix} \right] ^{\mathrm{H}}$ and
\begin{equation*}
  \widetilde{\mathbf{q}}_k=\left[ \begin{array}{c}
    \delta _{\mathrm{D},k}\mathbf{\Xi }_k( \widehat{\mathbf{h}}_{\mathrm{D},k}+\widehat{\mathbf{H}}_{k}^{\mathrm{H}}\boldsymbol{\theta } )\\
    \delta _{\mathrm{H},k}\mathrm{vec}^*( \boldsymbol{\theta }( \widehat{\mathbf{h}}_{\mathrm{D},k}^{\mathrm{H}}+\boldsymbol{\theta }^{\mathrm{H}}\widehat{\mathbf{H}}_k ) \mathbf{\Xi }_k )\\
  \end{array} \right].
\end{equation*}
The fourth component on the left of \eqref{eq:full-outage} is reformulated as
\begin{align*}
  &\bigtriangleup \mathbf{H}_{k}^{\mathrm{H}}\mathbf{\Xi }_k\bigtriangleup \mathbf{h}_{\mathrm{D},k}+2\mathfrak{R} \{ \boldsymbol{\theta }^{\mathrm{H}}\bigtriangleup \mathbf{H}_k\mathbf{\Xi }_k\bigtriangleup \mathbf{h}_{\mathrm{D},k} \} \\
  &  +\boldsymbol{\theta }^{\mathrm{H}}\bigtriangleup \mathbf{H}_k\mathbf{\Xi }_k\bigtriangleup \mathbf{H}_{k}^{\mathrm{H}}\boldsymbol{\theta }
\\
=& \thinspace\delta _{\mathrm{D},k}^{2}\mathbf{i}_{\mathrm{D},k}^{\mathrm{H}}\mathbf{\Xi }_k\mathbf{i}_{\mathrm{D},k}+2\mathfrak{R} \{ \bigtriangleup \mathbf{h}_{\mathrm{D},k}^{\mathrm{H}}( \mathbf{\Xi }_k\otimes \boldsymbol{\theta }^{\mathrm{T}} ) \mathrm{vec}\left( \bigtriangleup \mathbf{H}_{k}^{*} \right) \}
\\
& +\mathrm{vec}^{\mathrm{T}}\left( \bigtriangleup \mathbf{H}_k \right) \left( \mathbf{\Xi }_k\otimes \mathbf{\Theta }^{\mathrm{T}} \right) \mathrm{vec}( \bigtriangleup \mathbf{H}_{k}^{*} )
\\
= &\thinspace\delta _{\mathrm{D},k}^{2}\mathbf{i}_{\mathrm{D},k}^{\mathrm{H}}\mathbf{\Xi }_k\mathbf{i}_{\mathrm{D},k}+2\mathfrak{R} \{ \delta _{\mathrm{D},k}\delta _{\mathrm{H},k}\mathbf{i}_{\mathrm{D},k}^{\mathrm{H}}( \mathbf{\Xi }_k\otimes \boldsymbol{\theta }^{\mathrm{T}} ) \mathbf{i}_{\mathrm{H},k}^{*} \}
\\
& +\delta _{\mathrm{H},k}^{2}\mathbf{i}_{\mathrm{H},k}^{\mathrm{T}}( \mathbf{\Xi }_k\otimes \boldsymbol{\theta }^{\mathrm{T}} ) \mathbf{i}_{\mathrm{H},k}^{*} =\widetilde{\mathbf{i}}_{k}^{\mathrm{H}}\widetilde{\mathbf{Q}}_k\widetilde{\mathbf{i}}_k,
\end{align*}
where
\begin{equation*}
  \widetilde{\mathbf{Q}}_k=\left[ \begin{matrix}
    \mathbf{\Sigma }_{\mathrm{D},k}^{1/2}\mathbf{\Xi }_k\mathbf{\Sigma }_{\mathrm{D},k}^{1/2}&		\delta _{\mathrm{D},k}\delta _{\mathrm{H},k}(\mathbf{\Xi }_k\otimes \boldsymbol{\theta }^{\mathrm{T}})\\
    \delta _{\mathrm{D},k}\delta _{\mathrm{H},k}(\mathbf{\Xi }_k\otimes \boldsymbol{\theta }^*)&		\delta _{\mathrm{H},k}^{2}(\mathbf{\Xi }_k\otimes \mathbf{\Theta }^{\mathrm{T}})\\
  \end{matrix} \right].
\end{equation*}

Let $\widetilde{q}_k=( \widehat{\mathbf{h}}_{\mathrm{D},k}^{\mathrm{H}}+\boldsymbol{\theta }^{\mathrm{H}}\widehat{\mathbf{H}}_k ) \mathbf{\Xi }_k( \widehat{\mathbf{h}}_{\mathrm{D},k}+\widehat{\mathbf{H}}_{k}^{\mathrm{H}}\boldsymbol{\theta } )-\sigma _{k}^{2}$, the  OP \eqref{eq:full-outage} is then equivalent to
\begin{equation} \label{eq:full-outage-new}
  \mathrm{Pr}\{ \widetilde{\mathbf{i}}_{k}^{\mathrm{H}}\widetilde{\mathbf{Q}}_k\widetilde{\mathbf{i}}_k+2\mathfrak{R} ( \widetilde{\mathbf{q}}_{k}^{\mathrm{H}}\widetilde{\mathbf{i}}_k ) +\widetilde{q}_k\ge 0 \} \ge 1-\varepsilon _k.
\end{equation}

By using Lemma \ref{BTI} with the introduction of auxiliary variables $\widetilde{\mathbf{x}}=\left[ \widetilde{x}_1,\cdots ,\widetilde{x}_K \right] ^{\mathrm{T}}$ and $\widetilde{\mathbf{y}}=\left[ \widetilde{y}_1,\cdots ,\widetilde{y}_K \right] ^{\mathrm{T}}$, the approximation of the OP constraint for the $k$ user in \eqref{eq:full-outage-new} is expressed as
\begin{subequations} \label{eq:BTI-full}
  \begin{align}
    &\mathrm{Tr}\{ \widetilde{\mathbf{Q}}_k \} -\sqrt{2\ln\mathrm{(}1/\varepsilon _k)}\widetilde{x}_k+\ln\mathrm{(}\varepsilon _k)\widetilde{y}_k+\widetilde{q}_k\ge 0,  \label{eq:BTI-full-1} \\
    &\sqrt{||\widetilde{\mathbf{Q}}_k||_{F}^{2}+2||\widetilde{\mathbf{q}}_k|||^2}\le \widetilde{x}_k, \label{eq:BTI-full-2}  \\
    &\widetilde{y}_k\mathbf{I}+\widetilde{\mathbf{Q}}_k\succeq \mathbf{0},\widetilde{y}_k\ge 0. \label{eq:BTI-full-3}
  \end{align}
\end{subequations}

We simplify some terms in \eqref{eq:BTI-full} as shown below:
\begin{subequations} \label{eq:BTI-full-new}
  \begin{align}
    &\mathrm{Tr}\{ \widetilde{\mathbf{Q}}_k \} =\mathrm{Tr}\left\{ \left[ \begin{array}{c}
      \delta _{\mathrm{D},k}\mathbf{\Xi }_{k}^{1/2}\\
      \delta _{\mathrm{H},k}(\mathbf{\Xi }_{k}^{1/2}\otimes \boldsymbol{\theta }^*)\\
    \end{array} \right] \right. \nonumber  \\
    & \ \ \ \ \ \ \ \ \ \ \ \   \left.  \bullet
    \left[ \begin{matrix}
      \delta _{\mathrm{D},k}\mathbf{\Xi }_{k}^{1/2}&		\delta _{\mathrm{H},k}(\mathbf{\Xi }_{k}^{1/2}\otimes \boldsymbol{\theta }^{\mathrm{T}})\\
    \end{matrix} \right] \right\}  \label {eq:BTI-full-1-new}
    \\
    & \| \widetilde{\mathbf{Q}}_k \| _{F}^{2}=( \delta _{\mathrm{D},k}^{2}+\delta _{\mathrm{H},k}^{2}N ) ^2\mathrm{Tr}\{ \mathbf{\Xi }_k \} \| \mathbf{\Xi }_k \| _{F}^{2},  \label{eq:BTI-full-2-new} \\
    &\| \widetilde{\mathbf{q}}_k \| ^2=( \delta _{\mathrm{D},k}^{2}+\delta _{\mathrm{H},k}^{2}N ) \|( \widehat{\mathbf{h}}_{\mathrm{D},k}^{\mathrm{H}}+\boldsymbol{\theta }^{\mathrm{H}}\widehat{\mathbf{H}}_k ) \mathbf{\Xi }_k | _{2}^{2}, \label{eq:BTI-full-3-new}  \\
    &\widetilde{y}_k\mathbf{I}+\widetilde{\mathbf{Q}}_k\succeq \mathbf{0}\Rightarrow \widetilde{y}_k\mathbf{I}+\left( \delta _{\mathrm{D},k}^{2}+\delta _{\mathrm{H},k}^{2}N \right) \mathbf{\Xi }_k\succeq \mathbf{0}. \label{eq:BTI-full-4-new}
  \end{align}
\end{subequations}

The derivation leading to \eqref{eq:BTI-full-new} mirror those of \eqref{eq:BTI}.

Based on the previous results, $\mathcal{P}_{7}$ involving imperfect DCSIB and CBRUB can be expressed as
\begin{subequations}
  \begin{align}   \label{eq:Pro11}
    & \mathcal{P}_{11}:
    \underset{\mathbf{W},\boldsymbol{\theta },\widetilde{\mathbf{x}},\widetilde{\mathbf{y}}}{\min} \left\| \mathbf{W} \right\| _{F}^{2}     \\
    & \ \ \  \mathrm{s}.\mathrm{t}. \thinspace\thinspace \left(\delta _{\mathrm{D},k}^{2} + \delta _{\mathrm{H},k}^{2}N \right) \mathrm{Tr}\{ \mathbf{\Xi }_k \} -\sqrt{2\ln\mathrm{(}1/\varepsilon _k)}\widetilde{x}_k
     \nonumber  \\
   & \ \ \ \ \ \ \ \  -\ln\mathrm{(}1/\varepsilon _k)\widetilde{y}_k+\widetilde{q}_{k}\ge 0,\forall k\in \mathcal{K}, \label{eq:Pro11b} \\
    &  \ \ \ \ \ \ \ \    \left\| \begin{array}{c}
      ( \delta _{\mathrm{D},k}^{2}+\delta _{\mathrm{H},k}^{2}N ) \mathrm{vec}( \mathbf{\Xi }_k )\\
      \sqrt{2( \delta _{\mathrm{D},k}^{2}+\delta _{\mathrm{H},k}^{2}N )}\mathbf{\Xi }_k( \widehat{\mathbf{h}}_{\mathrm{D},k}+\widehat{\mathbf{H}}_{k}^{\mathrm{H}}\boldsymbol{\theta } )\\
    \end{array} \right\|\le \widetilde{x}_{k}, \nonumber
    \\
    & \ \ \ \ \ \ \ \ \   \forall k\in \mathcal{K}, \label{eq:Pro11c} \\
    & \ \ \ \ \ \ \ \ \  \widetilde{y}_{k}\mathbf{I}+\left(\delta _{\mathrm{D},k}^{2} + \delta _{\mathrm{H},k}^{2}N \right)\mathbf{\Xi }_k\succeq0,\widetilde{y}_{k}\geq0,\forall k\in\mathcal{K}, \label{eq:Pro11d} \\
    & \ \ \ \ \ \ \ \ \ \left| \theta _n \right|^2=1,\forall n\in \mathcal{N}. \label{eq:Pro11e}
  \end{align}
\end{subequations}

By comparing $\mathcal{P}_{11}$ with $\mathcal{P}_{8}$, we find that the former can be derived from the latter by substituting $\delta _{\mathrm{H},k}^{2}N$ with $\delta _{\mathrm{D},k}^{2}+\delta _{\mathrm{H},k}^{2}N$ and $\mathbf{h}_{\mathrm{D},k}$ with $\widehat{\mathbf{h}}_{\mathrm{D},k}$. Thus, solving  $\mathcal{P}_{11}$ follows the same approach as solving  $\mathcal{P}_{8}$.

\section{Convergence and Computational Complexity Analysis} \label{complexity}
In this section, we analyze the convergence and computational complexity of the proposed algorithms.

\vspace{-2mm}
\subsection{Convergence}
In this work, we introduce the BCSIE and the SCSIE models and consider both PCU and FCU scenarios. In the PCU scenario, $\mathcal{P}_3$ and $\mathcal{P}_5$ are solved alternately for the BCSIE model, while $\mathcal{P}_9$ and $\mathcal{P}_{10}$ are solved alternately for the SCSIE model. Furthermore, in the FCU scenario, the solution methods for the two optimization problems $\mathcal{P}_6$ and $\mathcal{P}_{11}$, based on the BCSIE and the SCSIE models respectively, are similar to those used for the corresponding optimization problems in the PCU scenario. In the PCU scenario, each subproblem $\mathcal{P}_3$ and $\mathcal{P}_5$ under the bounded CSI error model, and each subproblem $\mathcal{P}_8$ and $\mathcal{P}_{10}$ under the statistical CSI error model converges to their individual suboptimal solutions. In the FCU scenario, the convergence of each subproblem is similar to that in the PCU scenario.

For the convergence analysis of the BCSIE and the SCSIE models in the PCU scenario, we have:
\begin{align}
  \operatorname{Power}_b\left(\mathbf{W}^{(l+1)},\boldsymbol{\theta}^{(l+1)}\right)&\overset{(a_1)}{\operatorname*{\leq}}\operatorname{Power}_b\left(\mathbf{W}^{(l)},\boldsymbol{\theta}^{(l+1)}\right) \nonumber \\ 
  &\overset{(b_1)}{\operatorname*{\leq}}\operatorname{Power}_b\left(\mathbf{W}^{(l)},\boldsymbol{\theta}^{(l)}\right), \\
  \operatorname{Power}_s\left(\mathbf{W}^{(l+1)},\boldsymbol{\theta}^{(l+1)}\right)&\overset{(c_1)}{\operatorname*{\leq}}\operatorname{Power}_s\left(\mathbf{W}^{(l)},\boldsymbol{\theta}^{(l+1)}\right) \nonumber \\ 
  &\overset{(d_1)}{\operatorname*{\leq}}\operatorname{Power}_s\left(\mathbf{W}^{(l)},\boldsymbol{\theta}^{(l)}\right),
\end{align}
where $l$ denotes the iteration index. $\operatorname{Power}_b$ and $\operatorname{Power}_s$ represent the transmit power under the BCSIE and the SCSIE models, respectively. Here, $(a_1)$ and $(b_1)$ follow since the updates of $\mathbf{W}$ and $\boldsymbol{\theta} $ minimize transmit power when the other variables are fixed in the BCSIE model. Then, $(c_1)$ and $(d_1)$  are due to the updates of $\mathbf{W}$ and $\boldsymbol{\theta} $ to minimize transmit power when the other variables are fixed in the SCSIE model. Furthermore, the fact that $\left\|\mathbf{w}_k\right\|^2 \ge 0 $ ensures that the objective function has a lower bound, both algorithms converge to suboptimal solutions of the original problems $\mathcal{P}_1$ and $\mathcal{P}_7$ after several iterations.

For the convergence analysis of the BCSIE and the SCSIE models in the FCU scenario is similar to the analysis in the PCU scenario.

\subsection{Computational Complexity Analysis}
In this subsection, we analyze the computational complexity associated with the robust transmission design algorithms.  Additionally, each problem incorporates LMIs, second-order cone (SOC) constraints, and linear constraints, which can be addressed via CVX tools. As mentioned in \cite{r16}, ignoring the complexity of linear constraints, we compare the computational complexity of various  approaches in terms of their worst-case runtime, with the general expression given as:


\[
\mathcal{O}\left((\sum_{j=1}^{J}b_{j}+2I)^{1/2}n(n^{2}+\underbrace{n\sum_{j=1}^{J}b_{j}^{2}+\sum_{j=1}^{J}b^{3}}_{{\rm {\mathsf{\mathrm{due\thinspace\thinspace to\thinspace\thinspace LMI}}}}}+\underbrace{n\sum_{i=1}^{I}a_{i}^{2}}_{{\rm {due\thinspace\thinspace to\thinspace\thinspace SOC}}})\right),
\]
where $n$ denotes the number of variables, $J$ represents the number of LMIs of dimension $b_j$, and $I$ indicates the number of SOC with dimension $a_i$. Using the earlier expression as a foundation, we outline the corresponding  computational complexity at each iteration of the algorithm for both types of CSI error models are given as follows:

1) \textit {The BCSIE model algorithm in the PCU scenario:} For the $\mathcal{P}_1$, the approximate complexity of the subproblem $\mathcal{P}_3$ is $o_{\mathbf{W}}=\mathrm{O}([K(NM+K+M+1)]^{1/2}n_1[n_1^2+n_1K((NM+1)^2+(K+M)^2)+K((NM+1)^3+(K+M)^3)])$, where $n_1=MK$, and that of the subproblem $\mathcal{P}_5$ is $o_{\boldsymbol{\theta}}=\mathcal{O}([K(NM+1+K)+2N]^{1/2}n_2[n_2^2+n_2K((NM+1)^2+K^2)+ K((NM+1)^3+K^3)+n_2N])$, where $n_2=N$. Therefore, the approximate complexity of PCU scenario in per iteration of $\mathcal{P}_1$ is $o_{\mathbf{W}} + o_{\boldsymbol{\theta}}$.

2) \textit {The BCSIE model algorithm in the FCU scenario:} For the $\mathcal{P}_6$, the approximate complexity of the beamforming matrix subproblem is $o_{\mathbf{W}}=\mathcal{O}([K(NM+3M+K+1)]^{1/2}n_{1}[n_{1}^{2}+n_{1}K((NM+M+1)^{2}+(K+2M)^{2})+K((NM+M+1)^{3}+(K+2M)^{2})])$ with $n_{1}=MK$, and that of the reflection beamforming optimization subproblem is $o_{\boldsymbol{\theta}}=\mathcal{O}([K(MN+1+K)+2M]^{1/2}n_2[n_2^2+n_2K((MN+1)^2+K^2)+ K((MN+1)^3+K^3)+n_2M])$, where $n_2=N$. Therefore, the approximate complexity of FCU scenario in per iteration of  $\mathcal{P}_6$ is $o_{\mathbf{W}} + o_{\boldsymbol{\theta}}$.

3) \textit {The SCSIE model algorithm in the PCU scenario:} For the  $\mathcal{P}_7$, in this case, the complexities of solving the two subproblems are  $o_{\mathbf{W}}=\mathcal{O}([2K(M+1)]^{1/2}n_{1}[n_{1}^{2}+2n_{1}KM^{2}+2KM^{3}+nKM^{2}(M+1)^{2}])$ and $o_{\boldsymbol{\theta}}=\mathcal{O}([4K+N]^{1/2}n_{2}[n_{2}^{2}+n_{2}(K(N^{2}+(M+1)^{2})+N)])$, respectively, where $n_{1}=MK$ and $n_{2}=N$. Hence, the overall complexity is $o_{\mathbf{W}} + o_{\boldsymbol{\theta}}$.

4) \textit {The SCSIE model algorithm  in the FCU scenario:} In the FCU scenario with statistic CSI error model, the per iteration complexity approximation is identical to that in the PCU scenario, since only certain coefficients vary.

\begin{table}
\caption{System parameters}
\label{table-parameters}
\centering
\begin{tabular}{|c|c|}
\hline
Speed the HST  & $v = 360$ km/h \tabularnewline
\hline
Carrier frequency & $f_c = 200$ MHz \tabularnewline
\hline
Path loss exponents of each link  & $\beta _{\{ {\mathrm{D},k};\mathbf{G}; {\mathrm{R},k}\}} =\{3.7,3,2\} $ \tabularnewline
\hline
Noise power  & \makecell{$  -174~ {\rm{dBm/Hz}}$ \\ $+10\log _{10}B+10 $ dB} \tabularnewline
\hline
System band  & $B = 200$ MHz \tabularnewline
\hline
Maximum outage probabilities  & $\varepsilon_{1}=...=\varepsilon_{K}=\varepsilon=0.05$ \tabularnewline
\hline
Rician K-factor  & $\kappa = 3$ dB \tabularnewline
\hline
\end{tabular}
\end{table}

\section{Numerical Result} \label{nr}
In this section, we offer a numerical assessment of the performance of the proposed algorithm. For small-scale fading, the variance of $\mathrm{vec}( \bigtriangleup \mathbf{H}_k )$ and $\bigtriangleup \mathbf{h}_{\mathrm{D},k}$ in SCSIE model is defined as  $\delta_{\mathrm{H},k}^2  = \omega _{\mathrm{H}}^2  ||\mathrm{vec}( \bigtriangleup \widehat{\mathbf{H}}_{k} )||_{2}^{2}$  and $\delta_{\mathrm{D},k}^2 =  \omega _{\mathrm{D}}^2 || \widehat{\mathbf{h}}_{\mathrm{D},k} ||_{2}^{2}$, respectively. $\omega_{\mathrm{H}} \in [0,1)$ and $\omega_{\mathrm{D}} \in[0,1)$ denote the associated  CSI uncertainty. For the BCSIE model, the extent of the uncertainty areas are given by $\xi_{\mathrm{H},k}=\sqrt{\frac{\delta_{\mathrm{H},k}^{2}}{2}F_{2NM}^{-1}(1-\varepsilon) }$ and $\xi_{\mathrm{h},k}=\sqrt{\frac{\delta_{\mathrm{D},k}^{2}}{2}F_{2M}^{-1}(1-\varepsilon) }$, where $F_{2NM}^{-1}(\cdot)$ and $F_{2M}^{-1}(\cdot)$ correspond to the inverse cumulative distribution function of Chi-square distribution with degrees of freedom $2NM$ and $2M$, respectively. As outlined in \cite{c15}, the BCSIE model described above facilitates a meaningful comparison between the performance of robust designs for worst-case and outage-constrained. Some simulation parameters are detailed in Table~\ref{table-parameters}.

The results shown in Fi{}g.~{\ref{fig:2-2}} highlight the effectiveness of the proposed scheme by illustrating the convergence rates of all evaluated methods in different scenarios. We compare the following schemes: 1) Ideal phase shift: the value of $\theta$ maximizes $| ( \mathbf{h}_{\mathrm{D},k}^{\mathrm{H}}+\boldsymbol{\theta }^{\mathrm{H}}\mathbf{H}_k ) \mathbf{w}_k |$ \cite{c12}, and the optimization of $\mathbf{w}_k$ using the proposed method; 2) Discrete phase shift: we consider the discrete phase shift scheme described in \cite{c20}, and the optimization of $\mathbf{w}_k$ using the proposed method. It is observed that all algorithms exhibit rapid convergence, with just $7$ iterations being sufficient. Additionally, it can also be observed from Fi{}g.~{\ref{fig:2-2}} that the ideal phase shift algorithm consumes less power than both the proposed algorithm and the discrete phase shift algorithm, and the proposed algorithm consumes less power than the discrete phase shift algorithm. Furthermore, the ideal phase shift algorithm converges faster than both the proposed algorithm and the discrete phase shift algorithm, while the proposed algorithm converges faster than the discrete phase shift algorithm. For the discrete phase shift algorithm, the speed of convergence improves the number of quantization bits.
\begin{figure}[!t]
  \centering
  \subfigure[In the PCU scenario, where $\omega_{\mathrm{H}} = 0.01$]{
  \begin{minipage}[t]{\linewidth}
  \centering
  \includegraphics[scale=0.515]{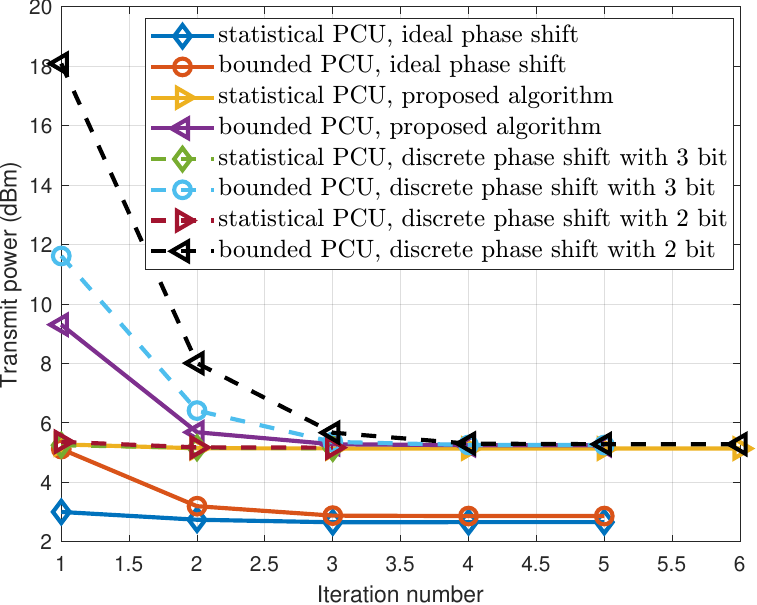}
  \label{fig:2-2a}
  \end{minipage}
  }
  \subfigure[In the FCU scenario, where $\{\omega_{\mathrm{H}},\omega_{\mathrm{D}}\}=\{0.01,0.02\}$]{
  \begin{minipage}[t]{\linewidth}
  \centering
  \includegraphics[scale=0.515]{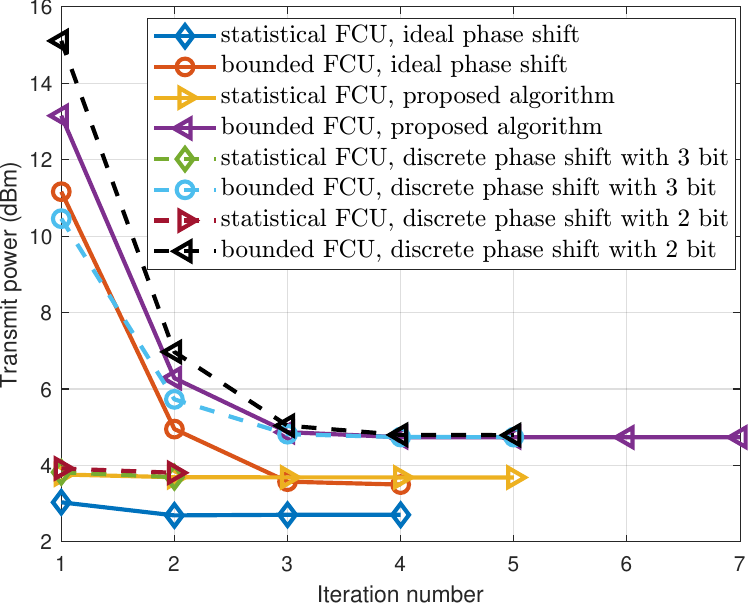}
  \label{fig:2-2b}
  \end{minipage}
  }
  \centering
  \caption{Transmit power versus the iteration numbers for each algorithm in different scenarios with $R=3$ bit/s/Hz, $K = 3$ and $M=N=3$.}
  \label{fig:2-2}
\end{figure}

The results shown in Fi{}g.~{\ref{fig:2}} highlight the effectiveness of the proposed scheme under both PCU and FCU scenarios, based on BCSIE and SCSIE models, and demonstrate the convergence speed of all evaluated methods. In this setup, $R=3$ bit/s/Hz, and $\{\omega_{\mathrm{H}},\omega_{\mathrm{D}}\}=\{0.01,0.02\}$. It is observed that all algorithms exhibit rapid convergence, with just $7$ iterations being sufficient. Additionally, the speed of convergence improves as the number of antennas increases. Furthermore, the algorithms based on the statistical error model converge more quickly than those with the use of BCSIE model.
\begin{figure}[!t]
  \centering
  {\includegraphics[scale=0.515]{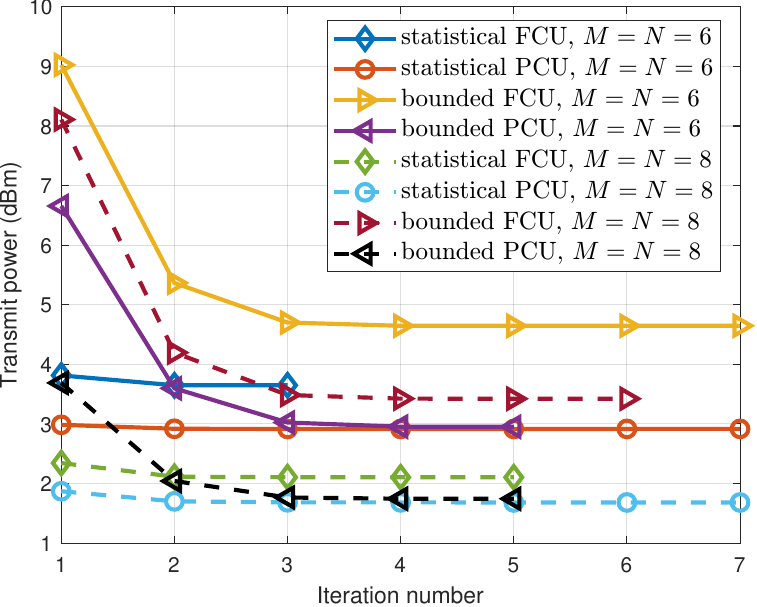}}
  \caption{Transmit power versus the iteration numbers of different algorithms, when $K=3$ and $\{\omega_{\mathrm{H}},\omega_{\mathrm{D}}\}=\{0.01,0.02\}$.}
  \label{fig:2}
\end{figure}

\begin{figure}[!t]
  \centering
  {\includegraphics[scale=0.515]{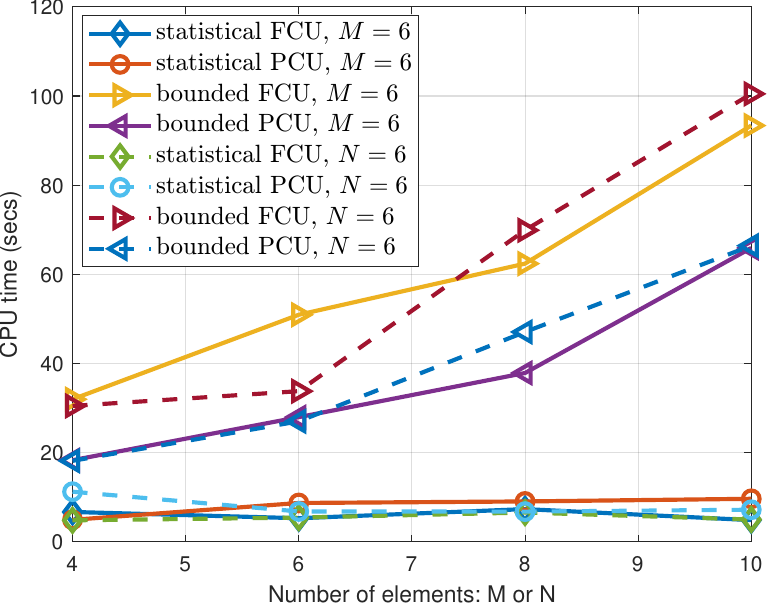}}
  \caption{Average CPU time versus the number of the BS antennas $M$ and the elements of the RIS $N$, when $K=2$ and $\{\omega_{\mathrm{H}},\omega_{\mathrm{D}}\}=\{0.01,0.02\}$.}
  \label{fig:3}
\end{figure}
Fi{}g.~{\ref{fig:3}} shows the comparison of the average (central processing unit) CPU running time for the proposed methods versus the number of the BS antennas $M$ and the  elements of the RIS $N$. The experiments were performed on a computer with a 3.3 GHz AMD Ryzen 9 5900HX CPU and 16 GB RAM. The parameters used are $K=2$, $R=2$ bit/s/Hz, and $\{\omega_{\mathrm{H}},\omega_{\mathrm{D}}\}=\{0.01,0.02\}$. Firstly, it is clear that the robust algorithms under the SCSIE model need much less CPU running time compared to the scenario with the BCSIE model. This is because the worst-case algorithms involve large-dimensional LMI, which increase the computational complexity. Secondly, the bounded FCU method consumes more CPU time than the bounded PCU method, as the DCSIB error $\bigtriangleup\mathbf{h}_{\mathbf{D},k}$ increases the LMIs dimension. Finally, when $N=6$, the CPU time for the OP constrained algorithm in both scenarios are nearly identical, since no additional complexity arises from considering the extra DCSIB error.

\begin{figure}[!t]
  \centering
  {\includegraphics[scale=0.515]{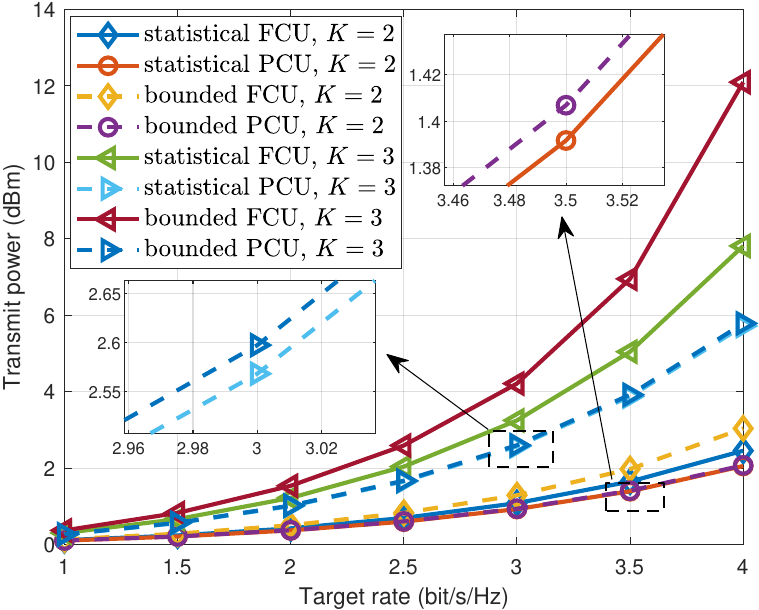}}
  \caption{Transmit power versus the target rate $R$ under $M=N=6$ and $\{\omega_{\mathrm{H}},\omega_{\mathrm{D}}\}=\{0.01,0.02\}$.}
  \label{fig:4}
\end{figure}
In Fi{}g.~{\ref{fig:4}}, we examine the impact of users' target rate requirements on the transmit power in RIS-assisted HST communication coverage enhancement under different CSI error models. The system parameters are specified as $K=\{2,3\}$, $N=M=6$, $\{\omega_{\mathrm{D}},\omega_{\mathrm{D}}\}=\{0.01,0.02\}$. It is observed that the transmit power grows as the target rate for both two scenarios and CSI error models. Furthermore, the transmit power required by the robust design for the worst-case algorithms exceeds that the robust design for the outage-constrained algorithms. This is because the worst-case optimization is a more conservative approach, necessitating higher transmit power to guarantee that each user's achievable rate meets the target rate under the worst-case CSI error realization.

\begin{figure}[!t]   
  \centering
  {\includegraphics[scale=0.515]{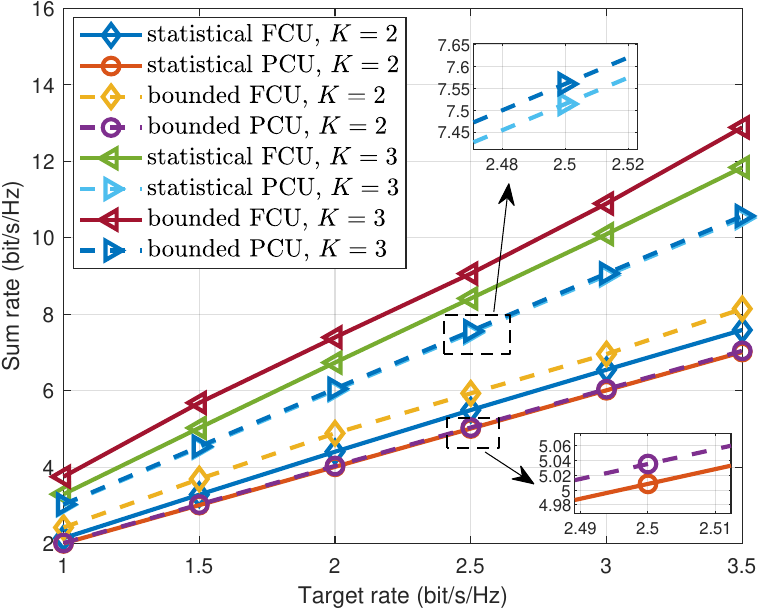}}
  \caption{Sum rate versus the target rate $R$ under $M=N=6$ and $\{\omega_{\mathrm{H}},\omega_{\mathrm{D}}\}=\{0.01,0.02\}$.}
  \label{fig:5}
\end{figure}
In Fig.~\ref{fig:5}, we show the sum rate of the RIS-aided HST communication coverage enhancement versus the target rate requirements of users under various CSI error models. Here, we set $N=M=6$, $K=\{2,3\}$, $\{\omega_{\mathrm{H}},\omega_{\mathrm{D}}\}=\{0.01,0.02\}$. It is observed that the sum rate increases with the target rate for both channel uncertainty scenarios and both CSI error.  In addition, it is also seen that the sum rate of the robust design algorithms with bounded CSI error model is larger than that of the robust design algorithms with statistical CSI error model. This is because of the transmit power of the worst-case robust design algorithms is larger than that of the outage constrained robust design algorithms from Fi{}g.~{\ref{fig:4}}, resulting the rate of each user in the worst-case is increased, thereby increasing the system sum rate.

In the following, we analyze the effect of CSI accuracy on system performance. We choose robust beamforming design algorithms subject to OP constraints, as the computational complexity of the worst-case conditions becomes prohibitively high with numerous antennas.

\begin{figure}[!t]
	\centering
	{\includegraphics[scale=0.515]{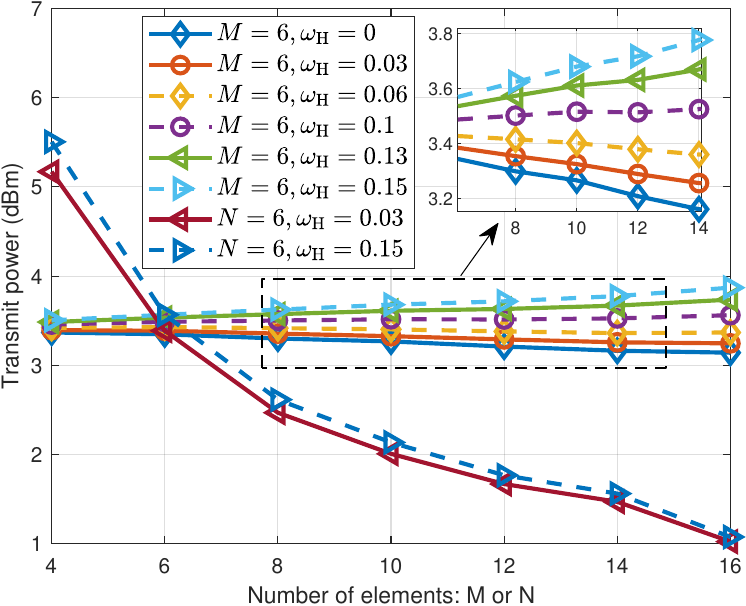}}
	\caption{Transmit power versus the number of the BS antennas $M$ and the elements of the RIS $N$ in the PCU scenario, when $K=3$.}
  \label{fig:6}
\end{figure}
Fi{}g.~{\ref{fig:6}} illustrates how the transmit power varies with $M$ or $N$ under the imperfect CBRUB, i.e., $\omega_{\mathrm{D}}=0$, and it is constructed using the channel realizations where feasible solutions are achieved at $M =16$ or $N =16$. The system parameters are set $K=3$ and $R=3$ bit/s/Hz. First, we analyze the scenario where the number of transmit antennas is fixed at $M = 6$.  In the case of $\omega_{\mathrm{H}} = 0$, representing the perfect CBRUB scenario, the transmit power decreases with reflection elements. For small values of $\omega_{\mathrm{H}}$, such as $\omega_{\mathrm{H}}=\{0.03,0.06\}$, the transmit power decreases as an increasing number of reflection elements, though it remains higher than in the perfect CBRUB case.  This is due to the need for the BS to boost the power to counteract the rate decline induced by the CBRUB error. Nevertheless, when $\omega_{\mathrm{H}}$ rises to $0.1$ or larger, the transmit power begins to increase as the number of RIS elements increase. This is because, although the additional reflection elements increase the beamforming gain and reduce power consumption, they also lead to larger channel estimation errors, requiring more power to compensate. As a result, when the CBRUB error becomes small, the advantages of increasing $N$, outweigh the shortcomings, and vice versa. Hence, choosing the right number of RIS elements and ensuring accurate CBRUB estimation are crucial for maximizing the RIS benefits.

In contrast, when the number of RIS elements remains fixed, the transmit power consumption reduces as the number of BS antennas increases, even in cases where the CBRUB error is high, such as $\omega_{\mathrm{H}}=0.15$. This is because a larger number of antennas provides more degrees of freedom, allowing for better optimization of the active beamforming vector at the BS, which helps mitigate the channel estimation errors.

\begin{figure}[!t]
  \centering
  \subfigure[$\omega_{\mathrm{H}} = 0.05$]{
  \begin{minipage}[t]{0.5\linewidth} %
  \centering
  \includegraphics[width=0.9\textwidth,height=0.7\textwidth]{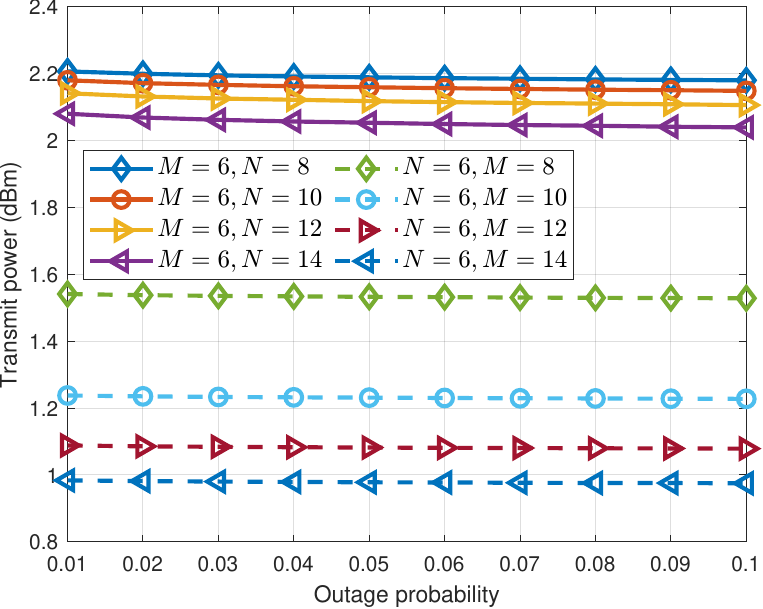}
  \label{fig:7a}
  \end{minipage}%
  }%
  \subfigure[$\omega_{\mathrm{H}} = 0.15$]{
  \begin{minipage}[t]{0.5\linewidth}
  \centering
  \includegraphics[width=0.9\textwidth,height=0.7\textwidth]{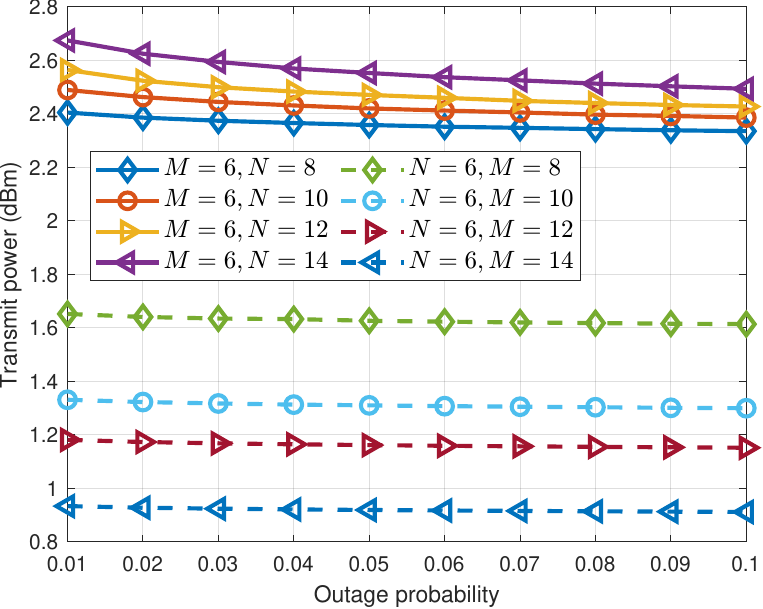}
  \label{fig:7b}
  \end{minipage}
  }
  \centering
  \caption{Transmit power versus the OP in the PCU scenario, when $K=3$.}
  \label{fig:7}
\end{figure}
Fi{}g.~{\ref{fig:7}} shows the transmit power versus the OP when only the CBIUT is imperfect. We assume $K = 3$ with $R = 3$ bit/s/Hz. As shown in Fi{}g.~{\ref{fig:7a}}, the transmit power decreases as the OP increases, which indicates that relaxing the OP criterion requires less transmit power. Additionally, comparing the impact of increasing the number of RIS reflection elements with increasing the number of antennas at the BS, the latter shows a more significant effect on reducing transmit power consumption. The results in Fig.~\ref{fig:7b} exhibit a similar trend to Fig.~\ref{fig:7a}, but with slight differences. Specifically, in Fig.~\ref{fig:7a}, when the number of RIS elements increases, the transmit power decreases for a smaller $\omega_{\mathrm{H}}$ value, e.g., $\omega_{\mathrm{H}} = 0.05$ and a fixed number of transmit antennas $M$, but increases when $\omega_{\mathrm{H}} = 0.15$. This trend is explained in Fig.~\ref{fig:6}.

\begin{figure}[!t]
  \centering
  {\includegraphics[scale=0.515]{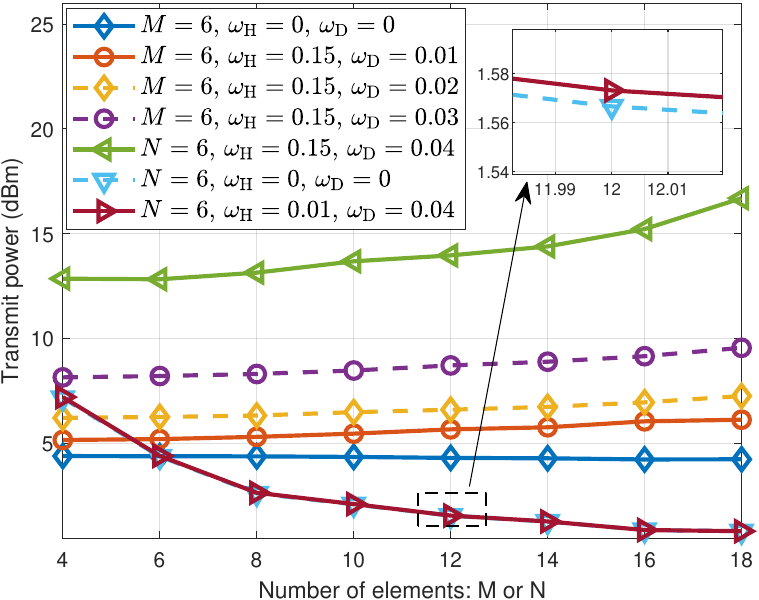}}
  \caption{Transmit power versus the number of the BS antennas $M$ and the elements of the RIS $N$ in the FCU scenario, when $K=3$.}
  \label{fig:8}
\end{figure}
Fi{}g.~{\ref{fig:8}} illustrates the transmit power versus $N$ or $M$ when both DCSIB and CBRUB are imperfect. The case where $\omega_{\mathrm{H}} = \omega_{\mathrm{D}} = 0 $ corresponds to the perfect DCSIB and  CBRUB case. Except for parameters $\omega_{\mathrm{H}}$ and $\omega_{\mathrm{D}}$, other simulation parameters match those in Fi{}g.~{\ref{fig:6}}. It is observed that increasing the number of BS  antennas result in a reduction in transmit power consumption, which remains unaffected by the DCSIB error $\omega_{\mathrm{D}}$, as seen in the curves for $N=6,\omega_{\mathrm{H}}=0.01,\omega_{\mathrm{D}}=\{0.01,0.02,0.3,0.04\}$.

\begin{figure}[!t]
	\centering
	{\includegraphics[scale=0.515]{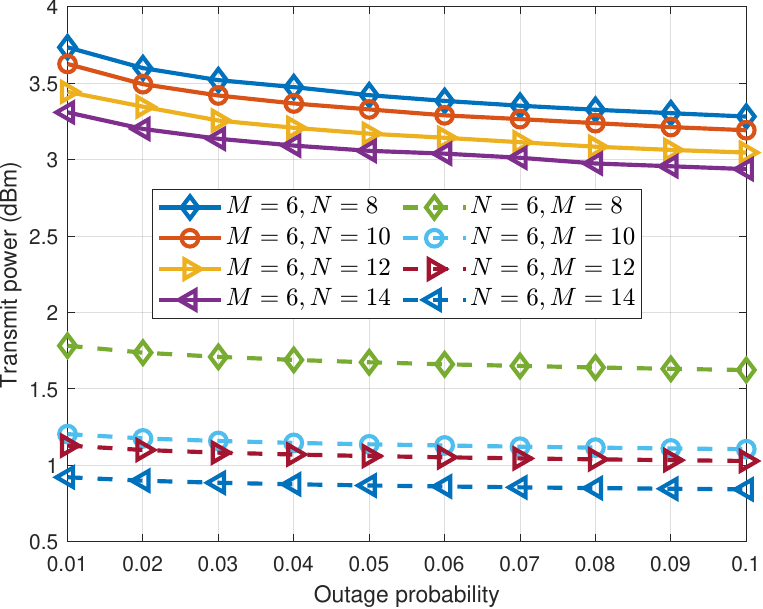}}
	\caption{Transmit power versus the OP in the FCU scenario, when $K=3$.}
  \label{fig:9}
\end{figure}
Fi{}g.~{\ref{fig:9}} presents the transmit power versus the OP when both DCSIB and CBRUB are imperfect. We set $K = 3$ and $R = 3$ bit/s/Hz, $\{\omega_{\mathrm{D}},\omega_{\mathrm{D}}\}=\{0.01,0.02\}$. The trend shown in this figure is consistent with that in Fig.~\ref{fig:7a}. However, the reduction in transmit power seen in Fig.~\ref{fig:9} is more pronounced than in Fig.~\ref{fig:7a}. This is due to the reason that the DCSIB error increases the overall channel estimation error, which in turn requires additional transmit power to mitigate its effects.

\section{Conclusions} \label{con}
In this study, we investigate robust beamforming and RIS phase shifts under imperfect CBRUB for RIS-aided MU-MISO HST communication coverage enhancement. We aim to minimize the BS's transmit power while adhering to rate constraints in the worst-case scenario under the BCSIE model and OP constraints in the SCSIE model, considering both PCU and FCU scenarios. To address CSI uncertainties in the BCSIE model, we employ the S-procedure to convert infinite inequality QoS constraints with unknown CSI errors into finite LMIs, and the penalty CCP is used to manage unit-modulus constraints. For the SCSIE model, CSI uncertainties are mitigated by applying the BTI. Additionally, unit-modulus constraints for RIS are transformed using the SDR technique. Simulation results indicate that the SCSIE model offers better performance, including  faster convergence, lower achievable, reduced complexity, and transmit power compared to the BCSIE model. Moreover, an increase in the number of RIS reflection elements can have an adverse effect on system performance when the CBRUB error is large. In our future work, we will carefully consider the impact of feedback delay, phase noise, and RIS hardware challenges on the performance of RIS-assisted HST communication systems, and propose low-complexity algorithms to further improve system performance. Furthermore, we will also study the performance analysis of RIS-assisted HST communication systems under non-stationary channel conditions.  When analyzing CSI errors, we will further comprehensively consider the impact of dynamics in HST communication on system performance.

\appendices

\section{The proof of Lemma \ref{le1}} \label{appa}
Let $x$ represent a complex scalar variable. The first-order Taylor inequality can be written as
\begin{equation} \label{eq:A1}
  | x |^2\ge 2\mathfrak{R} \{ x^{*,\left(n\right)}x \} -x^{*,(n)}x^{(n)},
\end{equation}
for any fixed point $x^{(n)}$. Substituting $x$ and $x^{\left(n\right)}$ into \eqref{eq:A1} with $( \mathbf{h}_{\mathrm{D},k}^{\mathrm{H}}+\boldsymbol{\theta }^{\mathrm{H}}\mathbf{H}_k ) \mathbf{w}_k$ and $( \mathbf{h}_{\mathrm{D},k}^{\mathrm{H}}+\boldsymbol{\theta }^{\left( n \right) ,\mathrm{H}}\mathbf{H}_k ) \mathbf{w}_{k}^{\left( n \right)}$, respectively, leads to the following expression:
\begin{align} \label{eq:A2}
&| ( \mathbf{h}_{\mathrm{D},k}^{\mathrm{H}}+\boldsymbol{\theta }^{\mathrm{H}}\mathbf{H}_k ) \mathbf{w}_k |^2  \nonumber \\  
&\ge 2\mathfrak{R} \{ ( \mathbf{h}_{\mathrm{D},k}^{\mathrm{H}}+\boldsymbol{\theta }^{\left(n\right),\mathrm{H}}\mathbf{H}_k ) \mathbf{w}_{k}^{\left(n\right)}\mathbf{w}_{k}^{\mathrm{H}}( \mathbf{h}_{\mathrm{D},k}^{\mathrm{H}}+\mathbf{H}_{k}^{\mathrm{H}}\boldsymbol{\theta } ) \}
\nonumber \\ 
 & \ \ \ -( \mathbf{h}_{\mathrm{D},k}^{\mathrm{H}}+\boldsymbol{\theta }^{\left(n\right),\mathrm{H}}\mathbf{H}_k ) \mathbf{w}_{k}^{\left(n\right)}\mathbf{w}_{k}^{\left(n\right),\mathrm{H}}( \mathbf{h}_{\mathrm{D},k}^{\mathrm{H}}+\mathbf{H}_{k}^{\mathrm{H}}\boldsymbol{\theta }^{\left(n\right)}).
\end{align}

By replacing $\mathbf{H}_k$ with $\widehat{\mathbf{H}}_k+\bigtriangleup $ on the right of \eqref{eq:A2} and developing it using established numerical identities, i.e., $\mathrm{Tr(}\mathbf{A}^{\mathrm{H}}\mathbf{B})=\mathrm{vec}^{\mathrm{H}}(\mathbf{A})\mathrm{vec(}\mathbf{B})$ and $\mathrm{Tr(}\mathbf{ABCD})=(\mathrm{vec}^\mathrm{T}(\mathbf{D}))^{\mathrm{T}}(\mathbf{C}^{\mathrm{T}}\otimes \mathbf{A})\mathrm{vec(}\mathbf{B})$ \cite{c18}, the expression in \eqref{eq:le1} can be derived. The proof has been concluded.

\section{The proof of Lemma \ref{lemma-full}} \label{appb}
The lower bound in \eqref{eq:lemma-full} can also be derived from \eqref{eq:A2} under the PCU. More specifically, by replacing  $\mathbf{h}_{\mathrm{D},k}=\widehat{\mathbf{h}}_{\mathrm{D},k}+\bigtriangleup \mathbf{h}_{\mathrm{D},k}$ and $\mathbf{H}_k=\widehat{\mathbf{H}}_k+\bigtriangleup \mathbf{H}_k$ into the first term on the right of \eqref{eq:A2}, we obtain \eqref{eq:B1}.
\begin{figure*}
  \begin{align} \label{eq:B1}
    & [ (\widehat{\mathbf{h}}_{\mathrm{D},k}^{\mathrm{H}}+\bigtriangleup \mathbf{h}_{\mathrm{D},k}^{\mathrm{H}})+\boldsymbol{\theta }^{(n),\mathrm{H}}(\widehat{\mathbf{H}}_k+\bigtriangleup \mathbf{H}_k) ] \mathbf{w}_{k}^{(n)}\mathbf{w}_{k}^{\mathrm{H}}[ (\widehat{\mathbf{h}}_{\mathrm{D},k}+\bigtriangleup \mathbf{h}_{\mathrm{D},k})+(\widehat{\mathbf{H}}_{k}^{\mathrm{H}}+\bigtriangleup \mathbf{H}_{k}^{\mathrm{H}})\boldsymbol{\theta } ] \nonumber
    \\
    = & \thinspace ( \widehat{\mathbf{h}}_{\mathrm{D},k}^{\mathrm{H}}+\boldsymbol{\theta }^{(n),\mathrm{H}}\widehat{\mathbf{H}}_k ) \mathbf{w}_{k}^{(n)}\mathbf{w}_{k}^{\mathrm{H}}( \widehat{\mathbf{h}}_{\mathrm{D},k}^{\mathrm{H}}+\widehat{\mathbf{H}}_{k}^{\mathrm{H}}\boldsymbol{\theta } ) +( \widehat{\mathbf{h}}_{\mathrm{D},k}^{\mathrm{H}}+\boldsymbol{\theta }^{(n),\mathrm{H}}\widehat{\mathbf{H}}_k ) \mathbf{w}_{k}^{(n)}\mathbf{w}_{k}^{\mathrm{H}}( \bigtriangleup \mathbf{h}_{\mathrm{D},k}+\bigtriangleup \mathbf{H}_{k}^{\mathrm{H}}\boldsymbol{\theta } ) \nonumber
    \\
    & +( \bigtriangleup \mathbf{h}_{\mathrm{D},k}^{\mathrm{H}}+\boldsymbol{\theta }^{(n),\mathrm{H}}\bigtriangleup \mathbf{H}_k ) \mathbf{w}_{k}^{(n)}\mathbf{w}_{k}^{\mathrm{H}}( \widehat{\mathbf{h}}_{\mathrm{D},k}+\widehat{\mathbf{H}}_{k}^{\mathrm{H}}\boldsymbol{\theta } )
   +( \bigtriangleup \mathbf{h}_{\mathrm{D},k}^{\mathrm{H}}+\boldsymbol{\theta }^{(n),\mathrm{H}}\bigtriangleup \mathbf{H}_k ) \mathbf{w}_{k}^{(n)}\mathbf{w}_{k}^{\mathrm{H}}( \bigtriangleup \mathbf{h}_{\mathrm{D},k}+\bigtriangleup \mathbf{H}_{k}^{\mathrm{H}}\boldsymbol{\theta } ) \nonumber
    \\
    =& \thinspace ( \widehat{\mathbf{h}}_{\mathrm{D},k}^{\mathrm{H}}+\boldsymbol{\theta }^{(n),\mathrm{H}}\widehat{\mathbf{H}}_k ) \mathbf{w}_{k}^{(n)}\mathbf{w}_{k}^{\mathrm{H}}( \widehat{\mathbf{h}}_{\mathrm{D},k}+\widehat{\mathbf{H}}_{k}^{\mathrm{H}}\boldsymbol{\theta } ) +( \widehat{\mathbf{h}}_{\mathrm{D},k}^{\mathrm{H}}+\boldsymbol{\theta }^{(n),\mathrm{H}}\widehat{\mathbf{H}}_k ) \mathbf{w}_{k}^{(n)}\mathbf{w}_{k}^{\mathrm{H}}\bigtriangleup \mathbf{h}_{\mathrm{D},k} \nonumber
    \\
    & +\mathrm{vec}^{\mathrm{H}}( \bigtriangleup \mathbf{H}_k ) \mathrm{vec}( \boldsymbol{\theta }( \widehat{\mathbf{h}}_{\mathrm{D},k}^{\mathrm{H}}+\boldsymbol{\theta }^{(n),\mathrm{H}}\widehat{\mathbf{H}}_k ) \mathbf{w}_{k}^{(n)}\mathbf{w}_{k}^{\mathrm{H}} ) +\bigtriangleup \mathbf{h}_{\mathrm{D},k}^{\mathrm{H}}\mathbf{w}_{k}^{(n)}\mathbf{w}_{k}^{\mathrm{H}}( \widehat{\mathbf{h}}_{\mathrm{D},k}+\widehat{\mathbf{H}}_{k}^{\mathrm{H}}\boldsymbol{\theta } )   +\bigtriangleup \mathbf{h}_{\mathrm{D},k}^{\mathrm{H}}\mathbf{w}_{k}^{(n)}\mathbf{w}_{k}^{\mathrm{H}}\bigtriangleup \mathbf{h}_{\mathrm{D},k} \nonumber   \\
    &  +\mathrm{vec}^{\mathrm{H}}( \boldsymbol{\theta }^{\left( n \right)}( \widehat{\mathbf{h}}_{\mathrm{D},k}+\boldsymbol{\theta }^{\mathrm{H}}\widehat{\mathbf{H}}_{k}^{\mathrm{H}} ) \mathbf{w}_{k}^{(n)}\mathbf{w}_{k}^{\mathrm{H}} ) \mathrm{vec}( \bigtriangleup \mathbf{H}_k )  +\mathrm{vec}^{\mathrm{H}}( \bigtriangleup \mathbf{H}_k ) ( \mathbf{w}_{k}^{*}\mathbf{w}_{k}^{(n),\mathrm{T}}\otimes \boldsymbol{\theta } ) \bigtriangleup \mathbf{h}_{\mathrm{D},k}^{*}    \nonumber    \\
    &  +\bigtriangleup \mathbf{h}_{\mathrm{D},k}^{\mathrm{T}}( \mathbf{w}_{k}^{*}\mathbf{w}_{k}^{(n),\mathrm{T}}\otimes \boldsymbol{\theta }^{(n),\mathrm{H}} ) \mathrm{vec}( \bigtriangleup \mathbf{H}_k ) +\mathrm{vec}^{\mathrm{H}}( \bigtriangleup \mathbf{H}_k ) ( \mathbf{w}_{k}^{*}\mathbf{w}_{k}^{(n),\mathrm{T}}\otimes \boldsymbol{\theta \theta }^{(n),\mathrm{H}} ) \mathrm{vec}( \bigtriangleup \mathbf{H}_k )    \nonumber    \\
    =& \thinspace \widetilde{\mathbf{i}}_{k}^{\mathrm{H}}\mathbf{J}_k\widetilde{\mathbf{i}}_k+\mathbf{j}_{1,k}^{\mathrm{H}}\widetilde{\mathbf{i}}_k+\widetilde{\mathbf{i}}_{k}^{\mathrm{H}}\mathbf{j}_{2,k}+j_k.
  \end{align}
  \hrule
\end{figure*}

Similarly, the two remaining components on the right of \eqref{eq:A2} under the assumption of PCU can be formulated as:
\begin{align}
  & ( \mathbf{h}_{\mathrm{D},k}^{\mathrm{H}}+\boldsymbol{\theta }^{\mathrm{H}}\mathbf{H}_k ) \mathbf{w}_k\mathbf{w}_{k}^{\left( n \right) ,\mathrm{H}}( \mathbf{h}_{\mathrm{D},k}+\mathbf{H}_{k}^{\mathrm{H}}\boldsymbol{\theta }^{\left( n \right)} ) \nonumber
  \\
  = & \thinspace\thinspace  \widetilde{\mathbf{i}}_{k}^{\mathrm{H}}\mathbf{J}_k\widetilde{\mathbf{i}}_k+\mathbf{j}_{2,k}^{\mathrm{H}}\widetilde{\mathbf{i}}_k+\widetilde{\mathbf{i}}_{k}^{\mathrm{H}}\mathbf{j}_{1,k}+j_{k}^{*} \nonumber
  \\
  & + ( \mathbf{h}_{\mathrm{D},k}^{\mathrm{H}}+\boldsymbol{\theta }^{\left( n \right) ,\mathrm{H}}\mathbf{H}_k ) \mathbf{w}_{k}^{\left( n \right)}\mathbf{w}_{k}^{\left( n \right) ,\mathrm{H}}( \mathbf{h}_{\mathrm{D},k}+\mathbf{H}_{k}^{\mathrm{H}}\boldsymbol{\theta }^{\left( n \right)} ) \nonumber
  \\
  = & \thinspace\thinspace  \widetilde{\mathbf{i}}_{k}^{\mathrm{H}}\mathbf{L}_k\widetilde{\mathbf{i}}_k+\mathbf{l}_{k}^{\mathrm{H}}\widetilde{\mathbf{i}}_k+\widetilde{\mathbf{i}}_{k}^{\mathrm{H}}\mathbf{l}_k+l_k.
\end{align}

Thus, the proof has been concluded.

\section{The proof of Theorem \ref{theorem-partical}} \label{appc}
Let $\widehat{\mathbf{\Psi}}^{\star}=[ \widehat{\mathbf{\Psi}}_{1}^{\star},\dots ,\widehat{\mathbf{\Psi}}_{K}^{\star} ] $ represent the optimal solution for the relaxed version of  $\mathcal{P}_9$, and specify the projection matrices as $\mathbf{U}_k=\widehat{\mathbf{\Psi}}_{k}^{\star \frac{1}{2}}\widehat{\mathbf{h}}_k\widehat{\mathbf{h}}_{k}^{\mathrm{H}}\widehat{\mathbf{\Psi}}_{k}^{\star \frac{1}{2}}/||\widehat{\mathbf{\Psi}}_{k}^{\star \frac{1}{2}}\widehat{\mathbf{h}}_k||^2, \forall k \in \mathcal{K}$, where $\widehat{\mathbf{h}}_k=\mathbf{h}_{\mathrm{D},k}+\mathbf{H}_{k}^{\mathrm{H}}\boldsymbol{\theta }$. We construct a rank-one solution $\widetilde{\mathbf{\Psi}}^{\star}=[ \widetilde{\mathbf{\Psi}}_{1}^{\star},\dots ,\widetilde{\mathbf{\Psi}}_{K}^{\star} ]$, where each sub-matrix is defined as
\begin{equation} \label{eq:rank1-1}
  \widetilde{\mathbf{\Psi}}_{k}^{\star}=\widehat{\mathbf{\Psi}}_{k}^{\star \frac{1}{2}}\mathbf{U}_k\widehat{\mathbf{\Psi}}_{k}^{\star \frac{1}{2}}.
\end{equation}

Initially, we evaluate the objective value of  $\mathcal{P}_9$ using the solution $\widetilde{\mathbf{\Psi}}_{k}^{\star}$:
\begin{equation} \label{eq:Pro9-check}
  \sum_{k=1}^K{\mathrm{Tr}\{ \widetilde{\mathbf{\Psi}}_{k}^{\star} -\widehat{\mathbf{\Psi}}_{k}^{\star} \}}=\sum_{k=1}^K{\mathrm{Tr}\{ \widehat{\mathbf{\Psi}}_{k}^{\star \frac{1}{2}}( \mathbf{U}_k-\mathbf{I} ) \widehat{\mathbf{\Psi}}_{k}^{\star \frac{1}{2}} \}}  \le 0,
\end{equation}
which indicates that the objective value obtained using  $\widetilde{\mathbf{\Psi}}_{k}^{\star}$ is smaller than the value derived from the optimal solution $\widehat{\mathbf{\Psi}}^{\star}$.

Then, given that directly confirming whether the constructed solution meets the constraints in \eqref{eq:Pro9b}-\eqref{eq:Pro9e} is computationally impractical, we now focus on the constraint \eqref{eq:Pro7b} in the original formulation of  $\mathcal{P}_7$.  From \eqref{eq:partial-outage}, we can obtain
\begin{align}  
  &\frac{\widehat{\mathbf{h}}_{k}^{\mathrm{H}}\widetilde{\mathbf{\Psi}}_{k}^{\star}\widehat{\mathbf{h}}_k}{(2^{R_{\mathrm{th}}}-1)}=\frac{|\widehat{\mathbf{h}}_{k}^{\mathrm{H}}\widehat{\mathbf{\Psi}}_{k}^{\star}\widehat{\mathbf{h}}_k|^2}{||\widehat{\mathbf{\Psi}}_{k}^{\star \frac{1}{2}}\widehat{\mathbf{h}}_k||^2(2^{R_{\mathrm{th}}}-1)}=\frac{\widehat{\mathbf{h}}_{k}^{\mathrm{H}}\widehat{\mathbf{\Psi}}_{k}^{\star}\widehat{\mathbf{h}}_k}{(2^{R_{\mathrm{th}}}-1)},  \label{eq:Pro9-check-1}  \\ 
  &\widehat{\mathbf{h}}_{k}^{\mathrm{H}}\widetilde{\mathbf{\Psi}}_{i}^{\star}\widehat{\mathbf{h}}_k =\widehat{\mathbf{h}}_{i}^{\mathrm{H}}\widehat{\mathbf{\Psi}}_{i}^{\star \frac{1}{2}}\frac{\widehat{\mathbf{\Psi}}_{i}^{\star \frac{1}{2}}\widehat{\mathbf{h}}_k\widehat{\mathbf{h}}_{k}^{\mathrm{H}}\widehat{\mathbf{\Psi}}_{i}^{\star \frac{1}{2}}}{||\widehat{\mathbf{\Psi}}_{i}^{\star \frac{1}{2}}\widehat{\mathbf{h}}_i||^2}\widehat{\mathbf{\Psi}}_{i}^{\star \frac{1}{2}}\widehat{\mathbf{h}}_i \nonumber \\
  & \ \ \ \ \ \ \ \  \ \ \   \le \lambda _{\max}( \widehat{\mathbf{\Psi}}_{i}^{\star \frac{1}{2}}\widehat{\mathbf{h}}_k\widehat{\mathbf{h}}_{k}^{\mathrm{H}}\widehat{\mathbf{\Psi}}_{i}^{\star \frac{1}{2}} ) =\widehat{\mathbf{h}}_{k}^{\mathrm{H}}\widehat{\mathbf{\Psi}}_{i}^{\star}\widehat{\mathbf{h}}_k.  \label{eq:Pro9-check-2}
\end{align}

By combining \eqref{eq:Pro9-check-1} and  \eqref{eq:Pro9-check-2}, we have
\begin{align}   \label{eq:Pro9-check-3}
  &\widehat{\mathbf{h}}_{k}^{\mathrm{H}}[\widetilde{\mathbf{\Psi}}_{k}^{\star}/(2^{R_{\mathrm{th}}}-1)-\sum_{i\ne k}^K{\widetilde{\mathbf{\Psi}}_{i}^{\star}]\widehat{\mathbf{h}}_k}
    \nonumber   \\
  & \ge \widehat{\mathbf{h}}_{k}^{\mathrm{H}}[\widehat{\mathbf{\Psi}}_{k}^{\star}/(2^{R_{\mathrm{th}}}-1)-\sum_{i\ne k}^K{\widehat{\mathbf{\Psi}}_{i}^{\star}]\widehat{\mathbf{h}}_k},
\end{align}
which indicates that the constructed solution $\widetilde{\mathbf{\Psi}}^{\star}$ satisfies the condition set by constraint \eqref{eq:Pro7b}, and also meets the constraints \eqref{eq:Pro9b}-\eqref{eq:Pro9e}.

With \eqref{eq:Pro9-check} and \eqref{eq:Pro9-check-3}, it follows that $\widetilde{\mathbf{\Psi}}^{\star}$ is also a feasibility solution for the relaxed version of  $\mathcal{P}_9$, having a rank of one. The proof has been concluded.

\end{document}